\documentclass[letter,12pt]{article}%
\usepackage{setspace}
\usepackage{booktabs}
\usepackage{amsmath,bm}
\usepackage{amsfonts}
\usepackage{amssymb}
\usepackage{graphicx}
\usepackage{rotating}
\usepackage{dsfont}
\usepackage[height=9in,left=1.25in,right=1.25in,bottom=1.25in, top = 1.25in]{geometry}
\usepackage[breaklinks=true,bookmarksopen=true,colorlinks=true,citecolor=blue]%
{hyperref}
\usepackage[authoryear]{natbib}
\usepackage{color}
\usepackage{colortbl}
\usepackage{paralist}
\usepackage{amsmath}%
\usepackage{mathtools}
\usepackage{tikz}
\usetikzlibrary{shapes.geometric, arrows}

\tikzstyle{startstop} = [rectangle, rounded corners, minimum width=3cm, minimum height=1cm,text centered, draw=black, fill=red!30]
\tikzstyle{process} = [rectangle, minimum width=3cm, minimum height=1cm, text centered, draw=black, fill=orange!30]
\tikzstyle{arrow} = [thick,->,>=stealth]
\providecommand{\U}[1]{\protect\rule{.1in}{.1in}}
\RequirePackage{hypernat}
\newtheorem{theorem}{Theorem}

\newtheorem{assumption}{Assumption}

\newtheorem{corollary}{Corollary}

\newtheorem{example}{Example}

\newtheorem{lemma}{Lemma}

\newtheorem{proposition}{Proposition}[section]

\newtheorem{manualtheoreminner}{Example}
\newenvironment{examplecont}[1]{%
  \renewcommand\themanualtheoreminner{#1}%
  \manualtheoreminner
}{\endmanualtheoreminner}
\newenvironment{proof}[1][Proof]{\noindent \textbf{#1:} }{\  \rule{0.5em}{0.5em}}

\newcommand{\floor}[1]{\lfloor #1 \rfloor}
\newcommand{\Ex}{\mathbb{E}}
\renewcommand{\U}{\mathbb{U}}
\newcommand{\ind}{\mathds{1}}

\newcommand{\T}[1]{\widetilde{#1}}

\renewcommand{\Pr}{\mathbb{P}}

\DeclarePairedDelimiter{\ceil}{\lceil}{\rceil}
\catcode`~=11
\newcommand{\urltilde}{\kern -.15em\lower .7ex\hbox{~}\kern .04em}
\catcode`~=13
\def \@seccntformat#1{\csname the#1\endcsname.\quad}
\makeatother
\numberwithin{equation}{section}
\hypersetup{pdftitle={Terschuur}, pdfsubject={Empirical Welfare Maximization}, pdfauthor={Joël Terschuur}, pdfkeywords={} }
\begin{document}
\title{\vspace{-2cm} Locally Robust Policy Learning: Inequality, Inequality of Opportunity and Intergenerational Mobility\thanks{Research funded by Ministerio de Ciencia e Innovaci\'{o}n,
grant ECO2017-86675-P, MCI/AEI/FEDER/UE, grant PGC 2018-096732-B-100, grant PID2021-127794NB-I00, grant PRE2020-092794.
Comunidad de Madrid, grants EPUC3M11 (VPRICIT) and H2019/HUM-589. Swiss National Science Foundation Grant 192580 `Bringing Empirical Welfare Maximization methods to practice'. I am grateful to Juan Carlos Escanciano, Toru Kitagawa and Aleksey Tetenov for their guidance and support, and to Miguel Ángel Delgado, Juan José Dolado, Ignacio Ortuño, Nazarii Salish, Jan Stuhler, Carlos Velasco and Daniel Wilhelm. \textit{Email address: joel.terschuur@tum.de}}}
\author{Joël Terschuur \\ \textit{Technical University of Munich, Germany}}
\date{\today \\
\href{https://raw.githubusercontent.com/joelters/website/gh-pages/assets/LRPL.pdf}{Click here for the latest version of the paper.}}
\maketitle

\begin{abstract}
    Policy makers need to decide whether to treat or not to treat heterogeneous individuals. The optimal treatment choice depends on the welfare function that the policy maker has in mind and it is referred to as the policy learning problem. I study a general setting for policy learning with semiparametric Social Welfare Functions (SWFs) that can be estimated by locally robust/orthogonal moments based on U-statistics. This rich class of SWFs substantially expands the setting in \cite{athey2021policy} and accommodates a wider range of distributional preferences. Three main applications of the general theory motivate the paper: (i) Inequality aware SWFs, (ii) Inequality of Opportunity aware SWFs and (iii) Intergenerational Mobility SWFs. I use the Panel Study of Income Dynamics (PSID) to assess the effect of attending preschool on adult earnings and estimate optimal policy rules based on parental years of education and parental income.
    \\
    \\
    \textbf{JEL\ Classification:} C13; C14; C21; D31; D63; I24
    \\
    \textbf{R package (forthcoming):} \url{https://joelters.github.io/home/code/}
\end{abstract}

\newpage

\section{Introduction}
Whenever a treatment has heterogeneous effects it is important to decide carefully who should be treated. In the simplest case where we care about the average outcome, no budgetary limits exist and treatment effects are positive for everyone, the best policy is to treat everyone. However, we might have a limited budget, distributional concerns or negative treatment effects. Then, it is important to decide whether \emph{to treat or not to treat} different individuals. This is the problem of policy learning. 

In economics we might want to know whether to provide training to the unemployed or design rules to assign conditional cash transfers. In business, we might want to know whether to provide a discount to a customer. Judges have to decide whether to release someone on parole. Schools might want to know whether to provide extra-curricular lessons to some students. Certain medicines might be beneficial for some but detrimental for others.

The inherent distributional conscerns in these examples are quite different. Hence, we need a framework accommodating different SWFs. The framework has to be general, but also needs to allow for certain statistical guarantees. I provide a framework to estimate optimal rules for a rich class of semiparametric SWFs, estimable by U-statistics. Examples include the average outcome, Inequality aware SWFs, Inequality of Opportunity (IOp) aware SWFs and Intergenerational Mobility (IGM) SWFs. 

To my knowledge, there is no prior work on IOp and IGM SWFs in the policy learning literature. IOp is the part of inequality explained by circumstances $X$ outside the control of the individual, e.g. sex, race or parental income. IOp SWFs do not penalize all inequality, only inequality explained by circumstances. Based on the seminal contributions in \cite{gaer1993}, \cite{fleurbaey1995equal} and \cite{roemer1998equality}, the IOp literature has focused on measuring IOp. A popular measure is the Gini of the best predictions (in mean squared error sense) of the outcome $Y$ given the circumstances $X$, i.e. $G(\gamma(X))$ where $\gamma(X) = \Ex[Y|X]$ and $G(Z)$ denotes the Gini index of the random variable $Z$. To accommodate a possibly high-dimensional set of circumstances, IOp literature has started using machine learners to predict (e.g \cite{brunori2019inequality}, \cite{brunori2019upward}, \cite{brunori2021roots}, \cite{brunori2021evolution}, \cite{rodriguez2021inequality}, \cite{carranza2022} or \cite{hufe2022fairness}). The bias-variance trade-off in the prediction allows for bias which can creep into the IOp estimator. \cite{escanciano2023machine} provide locally robust IOp estimators robust to such biases. I construct Neyman-orthogonal IOp aware SWFs.

Inequality SWFs have been studied in \cite{kasy2016partial}, \cite{kitagawa2021equality} or \cite{kock2023treatment}. A popular SWF for an outcome $Y$ is $W = \Ex[Y](1-G(Y))$. This SWF values the average outcome but penalizes high inequality. \cite{leqi2021median} propose to maximize average conditional quantiles. \cite{cui2024policy} focus on a conditional quantile of the treatment effect using partial identification. \cite{wang2018quantile} study quantile-optimal policies and adapt their theory to minimize Gini's mean difference. They use the U-statistics nature of the Gini mean difference and obtain asymptotic theory for a particular class of policy rules by using empirical U-process methods. I avoid U-processes theory by using a representation of U-statistics as sums-of-i.i.d. blocks introduced in \cite{hoeffding1963}. This representation is key in proving the main result of the paper.

While inequality aware SWFs look at the distribution of $Y$, IOp aware SWFs focus on the distribution of $\gamma(X)$. An IOp aware SWF is $W = \Ex[Y](1-G(\gamma(X)))$, which penalizes IOp. This SWF adds an extra nuisance parameter, $\gamma(X)$, on top of the conditional expectations/propensity scores needed to identify treatment effects. Policy learning with semiparametric SWFs, which directly depend on additional unknown functions, has been little explored, with the exception of \cite{leqi2021median} whose welfare depends on conditional quantiles.

IGM studies the relationship between child and parental outcomes. The Kendall-$\tau$ is a popular measure of mobility in the literature (see \cite{chetty2014land} or \cite{kitagawa2018measurement}). It looks at whether the parents of individual $i$ are richer than those of $j$ and $i$ is richer than $j$. An IGM aware SWF is $W = -|\tau - t|$ for some target $t \in [-1,1]$. For instance, we could decide the allocation of higher education scholarships to reduce dependence between parental and child's income.

The policy learning literature looks for optimal allocation rules $\pi$ mapping characteristics to binary treatment decisions. Optimal rules are searched within a class $\Pi$ of treatment rules to maximize welfare. Following \cite{manski2004statistical}, I search for optimal policies in $\Pi$ so as to minimize regret, i.e. the expected difference between the best possible welfare and the welfare evaluated at the estimated policy. Other relevant work includes \cite{dehejia2005program}, \cite{hirano2009asymptotics}, \cite{stoye2009minimax,stoye2012minimax}, \cite{chamberlain2011bayesian}, \cite{bhattacharya2012inferring}, \cite{tetenov2012statistical}, \cite{kasy2016partial}, \cite{kitagawa2018should,kitagawa2021equality}, \cite{athey2021policy}, \cite{} or \cite{zhou2023offline}.

Estimation of unknown functions in semiparametric SWFs challenges the statistical guarantees of estimated policy rules. This is due to slow convergence of non-parametric estimators, addressed in semiparametric methods through locally robust/orthogonal moments. These are moment conditions that identify the quantity of interest and allow for its estimation at $\sqrt{n}$ ($n$ is the sample size) rate. I expand previous work by considering any semiparametric SWF, possibly defined as a U-statistic, which can be estimated by locally robust/orthogonal scores. The main theoretical result provides an asymptotic upper bound to the regret of the estimated policy rule.

This paper is related to \cite{athey2021policy}, \cite{leqi2021median} and \cite{zhou2023offline} in making use of the semiparametric literature on locally robust/orthogonal scores (e.g. \cite{chernozhukov2022locally}) to obtain $\sqrt{n}$ rates of convergence even with nonparametric first steps. I build upon \cite{escanciano2023machine} to expand policy learning results to SWFs defined by U-statistics. \cite{athey2021policy} find rates of the regret that optimally depend on the complexity $\Pi$ and the sample size in observational settings where the propensity score is unknown. They do so for average-treatment-like SWFs. I generalize this setting by allowing general semiparametric SWFs, possibly defined as U-statistics.
 
Empirically, treatment allocation with inequality, IOp and IGM SWFs is hard. We need circumstances and parental income which are usually absent and to identify treatment effects. I look at the effect of preschool on adult earnings using the Panel Study of Income Dynamics (PSID) dataset. This empirical illustration has many advantages. Any variable that induces preschool attendance can be considered a circumstance under the (very reasonable) assumption that we cannot hold the kid responsible for these variables. Also, PSID has rich information on family background and it allows us to look at long-term outcomes. It also has limitations. Treatment is not randomly assigned, so I rely on the assumption of selection on observables. Preschool attendance is not a binary treatment since its quality varies. Furthermore, I have no information on the cost of treatment and allocating children to preschool based on their circumstances might not be enforceable or ethical. Observed preschool choices differ from estimated optimal rules, even when maximizing average outcomes, suggesting parents prioritize factors beyond future earnings.


The effect of preschool is heterogeneous. On average preschool has a positive effect on adult earnings but children with highly educated mothers and high parental income are negatively affected. This aligns with findings in psychology and economics (see \cite{fort2020cognitive}). These heterogeneous effects have different implications for different SWFs. I estimate optimal treatment rules based on parental income and mother's education. Inequality aware SWFs treat individuals with negative treatment effects since the decrease in inequality compensated the decrease in average earnings. The same happens with the IGM welfare which has no average motive at all. The additive and IOp estimated optimal policy rules coincide. This coincidence is specific to the heterogeneous treatment effects in the data and not a general result. In this empirical illustration, maximizing the average already decreases IOp drastically.

I introduce the welfare objects in Section \ref{sec_welfare_econ}. Section \ref{sec_pollearn_linorthscores} elaborates on the general theory for semiparametric SWFs which are linear on the distribution of the data (i.e. not U-statistics) and Section \ref{sec_pollearn_Ustats} generalized to SWFs possibly defined as U-statistics. Section \ref{sec_statguarantees} provides upper bounds on the regret and Section \ref{sec_empapp} deals with the empirical illustration. All proofs are in the Appendix. 

\section{Welfare economics for inequality, IOp and rank correlations}
\label{sec_welfare_econ}
The policy learning literature is at the intersection of welfare economics and econometrics. Before addressing the econometric problem, I introduce the key welfare objects of interest. For a continuous random outcome $Y_i \in \mathbb{R}^+$, the additive welfare is based on the average outcome: $W = \Ex[Y_i]$. Additive welfare does not care about distributional aspects other than the average. A first approach to include distributional concerns is to follow \cite{dalton1920measurement} and \cite{atkinson1970measurement} and consider increasing and concave transformations $u(\cdot)$ \footnote{With abuse of notation I call $W$ to all SWFs as they appear.}: $W = \Ex[u(Y_i)]$.

This SWF will rank two outcome distributions equally for all increasing and concave $u(\cdot)$ if the Lorenz curve of one of the distributions is everywhere above the Lorenz curve of the other distribution and has equal or higher mean; equivalently if one distribution second-order stochastically dominates the other. If we want to obtain a complete ordering we need to specify $u(\cdot)$ further. One popular choice is
\[
    u(y) = \begin{cases}
        \frac{y^{1-\theta}}{1-\theta} & \text{if } \theta \in (0,1) \\
        \log(y) & \text{if } \theta = 1,
    \end{cases}
\]
where $\theta$ captures the concavity of $u(\cdot)$ and can be interpreted as an inequality aversion parameter. I also focus on SWFs aware of Inequality of Opportunity (IOp). IOp is the part of total inequality which can be explained by circumstances, i.e. by variables that are outside the control of the individual such as parental education or parental income. Let $X_i \in \mathbb{R}^k$ be such a random vector of circumstances. Let also $\gamma(X_i) = \Ex[Y_i|X_i]$. By looking at the distribution of $\gamma(X_i)$ we get IOp averse SWFs: $W = \Ex[u(\gamma(X_i))]$.

If there is no IOp, circumstances are unable to predict the outcome and we have $\gamma(X_i) = \Ex[Y_i]$. In this case, $W = u(\Ex[Y_i])$ so we only care about the (transformed) average income. If we have maximum IOp, the outcome is a deterministic function of the circumstances and $\gamma(X_i) = Y_i$. Then, $W = \Ex[u(Y_i)]$. Since all inequality is IOp, we are back to the inequality averse SWF.

Alternatively, let $F_Y$ be the distribution of $Y$ and $F^{-1}_Y$ be the quantile function, for weights $w(\cdot)$, we might have
\[
    W = \int_0^1 F_Y^{-1}(\tau) w(\tau) d \tau.    
\]
This welfare has been used in \cite{mehran1976linear}, \cite{donaldson1980single}, \cite{weymark1981generalized}, \cite{donaldson1983ethically} or \cite{aaberge2021ranking}. Letting $w_k(\tau) = (k-1)(1-\tau)^{k-2}$ we get the extended Gini family of SWFs. I focus on $k = 3$, the standard Gini SWF which can be shown to be 
\begin{align*}
  W &= \Ex[Y_i](1 - G(Y_i)) = (1/2)\Ex[Y_i + Y_j - |Y_i - Y_j|],  
\end{align*}
where the second equality follows from writing the Gini of $Y_i$ as $G(Y_i) = \Ex[|Y_i-Y_j|]/\Ex[Y_i + Y_j]$ where $Y_j$ is an independent copy of $Y_i$. The welfare above is additive if there is no inequality and penalizes positive values of the Gini coefficient. If we only care about IOp, we can look at the distribution of $\gamma(X_i)$. In that case, we have
\begin{align*}
  W &= \Ex[\gamma(X_i)](1-G(\gamma(X_i))) = (1/2) \Ex[\gamma(X_i) + \gamma(X_j) - |\gamma(X_i) - \gamma(X_j)|].
\end{align*}
If $G(\gamma(X_i)) = 0$, we are back to the additive case. If there is full IOp, then $G(\gamma(X_i)) = G(Y_i)$ and we are back to the standard Gini SWF. I also consider the problem of intergenerational mobility. Let $X_{1i} \in \mathbb{R}$ be the parental outcome. A measure of association between $Y_i$ and $X_{1i}$ is the Kendall-$\tau$
\[
\tau = \Ex[sgn(Y_i-Y_j)sgn(X_{1i} - X_{1j})],
\]
where $sgn(a) = \ind(a > 0) - \ind(a< 0)$. This parameter is popular in the IGM literature (see \cite{chetty2014land} or \cite{kitagawa2018measurement}) where $X_{1i}$ is parental income and $Y_i$ is the child's income. It takes values between $1$ and $-1$. $\tau = 1$ means that whenever an individual has a higher income than another, she also has a higher parental income. $\tau = -1$ is the opposite. For some target Kendall-$\tau$ $t \in [-1,1]$ an IGM aware SWF is
\[
    W = -\biggl|\mathbb{E}\biggl[ sgn(Y_i-Y_j)sgn(X_{1i} - X_{1j}) \biggr] - t \biggr|.
\]
To my knowledge, this is a novel SWF. Note that it allows us to treat problems much more general than IGM. Setting $t = 0$, maximizing this SWF corresponds to allocating a treatment to minimize the dependence between two variables $Y_i$ and $X_{1i}$.

\section{Policy learning with general orthogonal scores}
\label{sec_pollearn_linorthscores}
Let $(Y_i(1),Y_i(0),D_i, X_i) \sim F_0$ where $(Y_i(1),Y_i(0)) \in \mathcal{Y} \times \mathcal{Y}$ are real-valued potential outcomes, i.e. $i$'s outcome under treatment and in the absence of treatment respectively. $D_i$ is a binary treatment and $X_i \in \mathcal{X}$ is a vector of pre-treatment covariates. Let $\gamma^{(j)}(X_i) = \mathbb{E}[Y_i(j)|X_i] \in \Gamma$ for $j = 0,1$ be potential predictions. We observe an i.i.d. sample $(Z_1,...,Z_n)$ with $Z_i = (Y_i,D_i,X_i) \in \mathcal{Z}$ and $Y_i = Y_i(1)D_i + Y_i(0)(1-D_i) \in \mathcal{Y}$. Let $\pi: \mathcal{X} \mapsto \{0,1\}$ be a binary treatment rule and $\Pi$ be a collection of such treatment rules. We are interested in choosing $\pi \in \Pi$ to maximize
\begin{align}
\label{eq_linwf}
W(\pi) &= \mathbb{E}[g(Y_i(1),X_i,\gamma^{(1)}) \pi(X_i) + 
g(Y_i(0),X_i,\gamma^{(0)}) (1-\pi(X_i))].
\end{align}
For additive welfare, $g(Y_i(j),X_i,\gamma^{(j)}) = Y_i(j)$ for $j = 0,1$. Importantly, $g$ can depend on possibly infinite-dimensional nuisance parameters $\gamma$. $\gamma$ is a conditional expectation throughout the paper but this framework can be extended to allow for much more general first steps such as conditional quantiles (see \cite{ichimura2022influence}).

\begin{example}[IOp Atkinson]
    For an inequality averse SWF we can use $W(\pi) = \Ex[u(Y_i(1))\pi(X_i) + u(Y_i(0))(1-\pi(X_i))]$ with $u(\cdot)$ a concave function and $X_i$ a vector of circumstances. For an IOp averse SWF we look at the distribution of $\gamma(X_i)$:
    \[
        W(\pi) = \Ex[u(\gamma^{(1)}(X_i))\pi(X_i) + u(\gamma^{(0)}(X_i))(1-\pi(X_i))].
    \] 
    This welfare has not been covered in the policy learning literature before.
    $\blacksquare$
\end{example}
(\ref{eq_linwf}) is not based on observables. To identify it, we need a sample from an experimental or observational experiment where the policy has already been implemented. Let $e(X_i) = \mathbb{P}(D_i  = 1 | X_i)$ be the propensity score. I assume that the following holds. 

\begin{assumption}
    \label{ass_ident}
    i) $(Y_i(1),Y_i(0)) \perp D_i | X_i$, ii) $e(x) \in [\kappa, 1 - \kappa]$ for some $\kappa \in (0,1/2]$.
\end{assumption}
There are two ways of identifying welfare, the direct method (DM) based on conditional expectations or Inverse Propensity Score Weighting (IPW). I use the DM approach. All results in the paper for the IPW approach are in Appendix \ref{app_ipw}. Let $\gamma(D_i,X_i) = \Ex[Y_i|D_i,X_i]$, $\gamma_j(X_i) = \gamma(j,X_i)$ for $j = 0, 1$ and $\varphi(D_i,X_i, \gamma) = \Ex[g(Y_i,X_i,\gamma)|D_i,X_i]$.
\begin{proposition}
    \label{prop_linid}
    Under Assumption \ref{ass_ident}, $W(\pi)$ is identified as
    \begin{align*}
        W(\pi) &= \Ex[\varphi(1,X_i,\gamma_1)\pi(X_i) + \varphi(0,X_i,\gamma_0)(1-\pi(X_i))].
    \end{align*}
\end{proposition}
If $g$ does not depend on potential outcomes directly, i.e. $g(u, X_i,\gamma^{(j)}) = g(t, X_i,\gamma^{(j)}) \equiv g(X_i,\gamma^{(j)})$ for all $u,t \in \mathcal{Y}$ then $\varphi(D_i, X_i,\gamma) = g(X_i,\gamma)$. This happens in all IOp examples. Hence, we can have either $\gamma$ or $(\gamma,\varphi)$ as nuisance parameters. To enjoy local robustness to first steps, I first need the following assumption

\begin{assumption}
    \label{ass_linlinearization}
    There exist $(\alpha_1,\alpha_0)$ such that for any $\tilde{\gamma}$ with $\Ex[\tilde{\gamma}(X_i)^2] < \infty$ and $j = 0,1$
    \[
        \frac{d}{d\tau} \Ex[\varphi(j,X_i,\bar{\gamma}_\tau)] \biggr|_{\tau = 0} = \frac{d}{d\tau} \Ex[\alpha_j(D_i,X_i)\bar{\gamma}_\tau(D_i,X_i)] \biggr|_{\tau = 0},
    \]
    where $\bar{\gamma}_\tau = \gamma + \tau \tilde{\gamma}$ and $\Ex[\alpha_j(D_i,X_i)^2] < \infty$.
\end{assumption}
This is a common assumption in the semiparametric literature (e.g. (4.1) in \cite{newey1994asymptotic}) and allows for $\varphi$ to depend non-linearly on $\gamma$, generalizing Assumption 1 in \cite{athey2021policy}. 

\begin{proposition}
    \label{prop_linorthscores}
    The orthogonal score is $\Gamma_i(\pi) = \Gamma_{1i}\pi(X_i) + \Gamma_{0i}(1-\pi(X_i))$, where
    \begin{align*}
    \Gamma_{1i} &= \varphi(1,X_i,\gamma) + \frac{D_i}{e(X_i)}(g(Y_i,X_i,\gamma_1) - \varphi(1,X_i,\gamma)) + \alpha_1(D_i,X_i)(Y_i-\gamma(D_i,X_i)), \\
    \Gamma_{0i} &= \varphi(0,X_i,\gamma) + \frac{1-D_i}{1-e(X_i)}(g(Y_i,X_i,\gamma_0) - \varphi(0,X_i,\gamma)) + \alpha_0(D_i,X_i)(Y_i-\gamma(D_i,X_i)).
    \end{align*}
\end{proposition}
Orthogonal scores are formed by identifying scores and correction terms for nuisance parameters $\varphi$ and $\gamma$. Whenever $g$ does not depend on the potential outcomes directly we have that $g(Y_i, X_i,\gamma_j) - \varphi(j, X_i,\gamma) = 0$ for $j = 0,1$. To estimate the welfare for a given $\pi \in \Pi$ we employ cross-fitting as in \cite{chernozhukov2022locally}. Let the data be split in $L$ groups $I_1,...,I_l$, then
\[
\hat{W}_n(\pi) = \frac{1}{n}\sum_{l = 1}^L \sum_{i \in I_l} \hat{\Gamma}_{1i,l}\pi(X_i) + \hat{\Gamma}_{0i,l}(1-\pi(X_i)),
\]
where 
\begin{align*} 
    \hat{\Gamma}_{1i,l} &= \hat{\varphi}_l(1,X_i,\hat{\gamma}_l) + \frac{D_i}{\hat{e}_l(X_i)}(Y_i - \hat{\varphi}_l(1,X_i,\hat{\gamma}_l)) + \hat{\alpha}_{1,l}(D_i,X_i)(Y_i-\hat{\gamma}_l(D_i,X_i)), \\
    \hat{\Gamma}_{0i,l} &= \hat{\varphi}_l(0,X_i,\hat{\gamma}_l) + \frac{1-D_i}{1-\hat{e}_l(X_i)}(Y_i - \hat{\varphi}_l(0,X_i,\hat{\gamma}_l)) + \hat{\alpha}_{0,l}(D_i,X_i)(Y_i-\hat{\gamma}_l(D_i,X_i)),
\end{align*}
and $(\hat{\varphi}_l,\hat{e}_l,\hat{\gamma}_l,\hat{\alpha}_{j,l})$, $j = 0,1$, are estimators of the nuisance functions which do not use observations in $I_l$.

\begin{examplecont}{1}[IOp Atkinson (cont.)]
For $\theta \in (0,1]$, let
\[
    U(\gamma(x)) = \begin{cases}
        \frac{\gamma(x)^{1-\theta}}{1-\theta} & \text{if } \theta \in (0,1) \\
        \log(\gamma(x)) & \text{if } \theta = 1.
    \end{cases}
\]
In this case, $g = U$. The orthogonal score for $\theta \in (0,1]$ is
\begin{align*}
    \Gamma_i(\pi) &= U(\gamma(1,X_i)) + \frac{\gamma(D_i,X_i)^{-\theta}D_i}{e(X_i)}(Y_i - \gamma(D_i,X_i))\pi(X_i) \\
    &+ U(\gamma(0,X_i)) + \frac{\gamma(D_i,X_i)^{-\theta}(1-D_i)}{1-e(X_i)}(Y_i - \gamma(D_i,X_i))(1-\pi(X_i)),
\end{align*}
i.e. $\alpha_1(D_i,X_i) = e(X_i)^{-1}\gamma(D_i,X_i)^{-\theta}D_i$ and $\alpha_0(D_i,X_i) = (1-e(X_i))^{-1}\gamma(D_i,X_i)^{-\theta}(1-D_i)$. 
$\blacksquare$
\end{examplecont}
The estimator of the optimal treatment rule among a class of rules $\Pi$ is $\hat{\pi} = \arg \max_{\pi \in \Pi} \hat{W}_n(\pi)$. 

\section{Policy learning with U-statistics}
\label{sec_pollearn_Ustats}
Let now $\pi_{ab}(X_i,X_j) = \ind(\pi(X_i) = a)\times\ind(\pi(X_j) = b)$ with $a,b \in \{0,1\}$. Now 
\begin{align}
    \label{eq_Uwf}
    W(\pi) &= \mathbb{E}\biggl[ \sum_{(a,b) \in \{0,1\}^2}g(Y_i(a),X_i,Y_j(b),X_j,\gamma^{(a)},\gamma^{(b)}) \pi_{ab}(X_i,X_j) \biggr].
\end{align}
Now, $W(\pi)$ depends on pairwise comparisons. Also, we are summing across $\{0,1\}^2$. This is because we have to take into account when both members of the pair are under treatment, or just one of them or none of them. Finally, we have $\pi_{ab}(X_i,X_j)$ instead of $\pi(X_i)$ since we need to account for when both members of the pair are allocated to treatment, just one of them or none of them. 

\begin{example}[Inequality]
\label{eg_inequality}
We can accommodate the standard Gini SWF with
\[
g(Y_i(a),Y_j(b)) = (1/2)(Y_i(a) + Y_j(b) - |Y_i(a) - Y_j(b)|).
\]
 $\blacksquare$
\end{example}

\begin{example}[Inequality of Opportunity IOp]
    \label{eg_iop}
        $\Ex[\gamma(X_i)](1-G(\gamma(X_i)))$ fits our setting by letting
        \[
        g(X_i,X_j,\gamma^{(a)},\gamma^{(b)}) = (1/2)(\gamma^{(a)}(X_i) + \gamma^{(b)}(X_j) - |\gamma^{(a)}(X_i) - \gamma^{(b)}(X_j)|).
        \]
        $\blacksquare$
\end{example}

\begin{example}[Kendal-$\tau$]
\label{eg_interm}
To allocate a treatment targeting a specific Kendall-$\tau$, say $t \in \mathbb{R}$, we have to extend our setting to transformations of the right-hand side of \ref{eq_Uwf}
\begin{align*}
g(Y_i(a),X_{1i},Y_j(b),X_{1j}) &= sgn(Y_i(a)-Y_j(b))sgn(X_{1i} - X_{1j}), \\
W(\pi) &= -\biggl|\mathbb{E}\biggl[ \sum_{(a,b) \in \{0,1\}^2}g(Y_i(a),X_{1i},Y_j(b),X_{1j}) \pi_{ab}(X_i,X_j) \biggr] - t \biggr|.
\end{align*}
$\blacksquare$
\end{example}

\noindent For $a,b \in \{0,1\}$ let now $\varphi(a,X_i,b,X_j,\gamma_a,\gamma_b) = \Ex[g(Y_i,X_i,Y_j,X_j,\gamma_a,\gamma_b)|D_i = a,X_i,D_j = b,X_j]$ and $e_{ab}(X_i,X_j) = e_a(X_i)e_b(X_j)$ where for $c \in \{0,1\}$, $e_c(X_i) = \Pr(D_i = c |X_i)$.

\begin{proposition}
    \label{prop_Uid}
    Under Assumption \ref{ass_ident}, $W(\pi)$ in (\ref{eq_Uwf}) is identified as
    \begin{align*}
        W(\pi) &= \Ex\biggl[\sum_{(a,b)\in \{0,1\}^2} \varphi(a,X_i,b,X_j,\gamma_a,\gamma_b)\pi_{ab}(X_i,X_j).\biggr],
    \end{align*}
\end{proposition}
Now I apply Proposition \ref{prop_Uid} to identify the welfare in each of our three main examples. 

\begin{examplecont}{2}[Inequality (cont.)]
    In this example, welfare is identified by
    \[
        W(\pi) = \mathbb{E}\biggl[\frac{1}{2} \sum_{(a,b) \in \{0,1\}^2} \Ex(Y_i + Y_j - |Y_i - Y_j| \mid D_i = a,X_i,D_j = b,X_j)\pi_{ab}(X_i,X_j)\biggr].    
    \]
    $\blacksquare$
\end{examplecont}

\begin{examplecont}{3}[IOp (cont.)]
    In this example, welfare is identified by
    \[
    W(\pi) = \frac{1}{2}\mathbb{E}\biggl[\sum_{(a,b)\in \{0,1\}^2}\biggl(\gamma_a(X_i) + \gamma_b(X_j) - |\gamma_a(X_i) - \gamma_b(X_j)|\biggr)\pi_{ab}(X_i,X_j)\biggr].
    \]
    $\blacksquare$
\end{examplecont}

\begin{examplecont}{4}[IGM (cont.)]
    In this example, welfare is identified by
    \[
    W(\pi) = - \biggl| \mathbb{E}\biggl[\frac{1}{2} \sum_{(a,b) \in \{0,1\}^2} \Ex(sgn(X_{1i} - X_{1j}) sgn(Y_i - Y_j) \mid D_i = a,X_i,D_j = b,X_j)\pi_{ab}(X_i,X_j)\biggr] - t \biggr|.
    \]
    $\blacksquare$
\end{examplecont}
To compute orthogonal scores we need to assume a linearization property as in Assumption \ref{ass_linlinearization}.

\begin{assumption}
    \label{ass_Ulinearization}
    There exist $\alpha_{ab,p}$, $P<\infty$, and $(c_{1p}, c_{2p})$ for $p = 1,...,P$, such that for all $(a,b) \in \{0,1\}^2$ the following linearization holds
    \[
    \frac{d}{d\tau} \mathbb{E}[\varphi(a,X_i,b,X_j,\bar{\gamma}_\tau)] = \frac{d}{d\tau} \mathbb{E}\biggl[\sum_{p=1}^P\alpha_{ab,p}^\gamma(D_i,X_i,D_j,X_j)(c_{1p}\bar{\gamma}_\tau(D_i,X_i) + c_{2p}\bar{\gamma}_\tau(D_j,X_j))\biggr],
    \]
    where $\bar{\gamma}_\tau$ is defined as in Assumption \ref{ass_linlinearization} and $\Ex[\alpha_j(D_i,X_i)^2] < \infty$.
\end{assumption}

\begin{proposition}
    \label{prop_Uorthscores}
    The orthogonal score is $\Gamma_{ij}(\pi) = \sum_{(a,b) \in \{0,1\}^2} \Gamma_{ij}^{ab}\pi_{ab}(X_i,X_j),$ with
    \begin{align*}
        \Gamma_{ij}^{ab} = 
        \varphi(a,X_i,b,X_j,\gamma_a,\gamma_b) &+ \phi_{ab}^\varphi(D_i,X_i,D_j,X_j,\varphi,\alpha^{\varphi}) + \phi_{ab}^\gamma(D_i,X_i,D_j,X_j,\gamma,\alpha^{\gamma}), \\
        \phi_{ab}^\gamma(D_i,X_i,D_j,X_j,e,\alpha^{\gamma}) &= \sum_{p=1}^P \alpha_{ab,p}^\gamma(D_i,X_i,D_j,X_j)(c_{1p} Y_i + c_{2p} Y_j - c_{1p}\gamma(D_i,X_i) - c_{2p} \gamma(D_j,X_j)), \\
        \phi_{ab}^\varphi(D_i,X_i,D_j,X_j,\varphi,\alpha^{m}) &=  \alpha_{ab}^\varphi(D_i,X_i,D_j,X_j)(g(Y_i,X_i,Y_j,X_j,\gamma_a,\gamma_b) - \varphi(D_i,X_i,D_j,X_j,\gamma_a,\gamma_b)),
    \end{align*}
    and $\alpha_{ab}^\varphi(D_i,X_i,D_j,X_j) = D_{ij}^{ab}/e_{ab}(X_i,X_j)$ and $D_{ij}^{ab} = \ind(D_i = a)\ind(D_j = b)$.
\end{proposition}

\begin{examplecont}{2}[Inequality (cont.)]
    In this example, we have that 
    \begin{align*}
        \Gamma_{ij}^{ab} &= \frac{1}{2}\Ex(Y_i + Y_j - |Y_i - Y_j| \mid D_i = a,X_i,D_j = b,X_j) \\
        &+ \frac{D_{ij}^{ab}}{2e_{ab}(X_i,X_j)}(Y_i + Y_j - |Y_i - Y_j| -\Ex(Y_i + Y_j - |Y_i - Y_j| \mid D_i = a,X_i,D_j = b,X_j)).
    \end{align*}
    $\blacksquare$
\end{examplecont}

 \begin{examplecont}{3}[IOp (cont.)]
    I give the orthogonal score for IOp as a Proposition.
    \begin{proposition}
        Assume for all $(a,b) \in \{0,1\}^2$ that either (i) $\mathbb{P}(\gamma_a(X_i) - \gamma_b(X_j) = 0) = 0$ or that (ii) $x_i \neq x_j \implies \gamma_a(X_i) - \gamma_b(X_j) \neq 0$ and let $\delta_{ij}^{ab} = sgn(\gamma_a(X_i) - \gamma_b(X_j))$, then
    \begin{align*}
        \Gamma_{ij}^{ab} &= \frac{1}{2}\biggl(\gamma_a(X_i) + \gamma_b(X_j) - |\gamma_a(X_i) - \gamma_b(X_j)| \\
        &+ \frac{\ind(D_i = a)}{e_a(X_i)}(1-\delta_{ij}^{ab})(Y_i - \gamma(D_i,X_i)) + \frac{\ind(D_j = b)}{e_b(X_j)}(1+\delta_{ij}^{ab})(Y_j - \gamma(D_j,X_j))\biggr).
    \end{align*}
    \end{proposition}
    \label{prop_ioporthscores}
    These assumptions deal with the point of non-differentiability of the absolute value. For a thorough discussion see \cite{escanciano2023machine}.
    $\blacksquare$
\end{examplecont}

\noindent I use a cross-fitting algorithm used in \cite{escanciano2023machine}. I split the pairs $\{(i,j) \in \{1,...,n\}^2: i < j\}$ in $L$ groups $I_1,...,I_l$, then
\begin{align}
    \label{eq_welfare_estimator}
\hat{W}_n(\pi) = \binom{n}{2}^{-1} \sum_{l = 1}^L \sum_{(i,j) \in I_l} \hat{\Gamma}_{ij,l}(\pi),
\end{align}
where $\hat{\Gamma}_{ij,l}$ is the same as $\Gamma_{ij}$ but with all nuisance parameters replaced by estimators which do not use observations in $I_l$. The estimator of the optimal treatment rule is
\[
  \hat{\pi} = \arg \max_{\pi \in \Pi} \hat{W}_n(\pi). 
\]
For the IGM example, the estimation is slightly different.

\begin{examplecont}{4}[IGM (cont.)]
    The orthogonal score is given by 
    \small
 \begin{align*}
     \Gamma_{ij}^{ab} &= \Ex(sgn(X_{1i} - X_{1j}) sgn(Y_i - Y_j) \mid D_i = a,X_i,D_j = b,X_j) \\
     &+ \frac{D_{ij}^{ab}}{e_{ab}(X_i,X_j)}(sgn(X_{1i} - X_{1j}) sgn(Y_i - Y_j) -\Ex(sgn(X_{1i} - X_{1j}) sgn(Y_i - Y_j) \mid D_i = a,X_i,D_j = b,X_j)).
 \end{align*}
 \normalsize
 The estimator of the welfare for a given $\pi \in \Pi$ and target $t$ is
 \begin{equation}
    \hat{W}_n(\pi) = -\biggl|\binom{n}{2}^{-1}\sum_{l = 1}^L \sum_{(i,j) \in I_l} \sum_{(a,b) \in \{0,1\}^2} \hat{\Gamma}_{ij,l}^{ab}\pi_{ab}(X_i,X_j) - t\biggr|.
    \label{eq_What_igm}
 \end{equation}
     $\blacksquare$
 \end{examplecont}

\section{Asymptotic statistical guarantees}
\label{sec_statguarantees}

Now it is useful to make clear the dependence of the scores $\Gamma_{ij}^{ab}$ on the data and the nuisance parameters. Hence, I let now $\Gamma_{ij}^{ab} = \psi_{ab}(Z_i,Z_j,\gamma,\varphi,\alpha)$, where
\[
    \psi_{ab}(Z_i,Z_j,\gamma,\varphi,\alpha) = \varphi(a,X_i,b,X_j,\gamma_a,\gamma_b) + \phi_{ab}^{\gamma}(Z_i,Z_j,\gamma,\alpha^{\gamma}) + \phi_{ab}^{\varphi}(Z_i,Z_j,\varphi,\alpha^{\varphi}).
\]
$\psi_{ab}$ is the sum of an identifying function ($\varphi_{ab}$) and correction terms ($(\phi_{ab}^{\gamma},\phi_{ab}^\varphi$) needed to achieve orthogonality. For a given treatment rule $\pi$, orthogonal scores are given by
\[
\Gamma_{ij}(\pi) = \sum_{(a,b) \in \{0,1\}^2}\psi_{ab}(Z_i,Z_j,\gamma,\varphi,\alpha)\pi_{ab}(X_i,X_j).
\]
This framework accommodates SWFs in Section \ref{sec_pollearn_linorthscores} if $\psi_{ab}$ does not depend on $Z_j$ and only depends on $a$ so that it can be written as $\psi_{a}(Z_i,\gamma,\varphi,\alpha)$ for $a \in \{0,1\}$. Hence, I only state conditions and results for SWFs in this Section. The IGM example does not fit in the general setting, however, results extend to this example by Corollary in Section C of the Online Appendix. Next, I give conditions on the convergence of the nuisance estimators and the complexity of $\Pi$ that allow me to prove asymptotic statistical guarantees for the estimator of treatment rules.

\subsection{Conditions on the nuisance parameter estimators}
I give high-level conditions for the estimators of nuisance parameters. 
\begin{assumption}
    \label{ass_nuisance_rates} $\mathbb{E}[|\psi(Z_{i},Z_{j},\gamma,\varphi,\alpha)|^{2}]<\infty$, $\omega \in \{\gamma,\varphi\}$ and for $(a,b) \in \{0,1\}^2$ 
    
    \begin{enumerate}
    [(i)]
    
    \item $n^{\lambda_\gamma}\sqrt{\Ex(|\varphi(a,X_{i},b,X_{j},\hat{\gamma}_{l})-\varphi(a,X_{i},b,X_{j},\gamma)|^{2})} = o(1) $ ;

    \item $n^{\lambda_\varphi}\sqrt{\Ex(|\hat{\varphi}_l(a,X_{i},b,X_{j},\gamma)-\varphi(a,X_{i},b,X_{j},\gamma)|^{2})} = o(1) $ ;

    \item $n^{\lambda_\gamma}\sqrt{\Ex(|\phi_{ab}^\gamma(Z_{i},Z_{j},\hat{\gamma}_l,\alpha^\gamma)-
    \phi_{ab}^{\gamma}(Z_{i},Z_{j},\gamma,\alpha^{\gamma})|^{2})}= o(1)$;

    \item $n^{\lambda_\varphi}\sqrt{\Ex(|\phi_{ab}^\varphi(Z_{i},Z_{j},\hat{\varphi}_l,\alpha^\varphi)-
    \phi_{ab}^{\varphi}(Z_{i},Z_{j},\varphi,\alpha^{\varphi})|^{2})} = o(1)$;
    
    \item $n^{\lambda_{\alpha}}\sqrt{\Ex(|\phi_{ab}^\omega(Z_{i},Z_{j},\omega,\hat{\alpha}_{l}^\omega)-
    \phi_{ab}^\omega(Z_{i},Z_{j},\omega,\alpha^{\omega})|^{2})} = o(1)$,
    \end{enumerate}
    where $1/4 < \lambda_\gamma, \lambda_\varphi, \lambda_\alpha$.
\end{assumption}
These are mean-square consistency conditions for $\hat{\gamma}_{l}$, $\hat{\varphi}_l$ and
$\hat{\alpha}_{l}$. Assumption \ref{ass_nuisance_rates} often follows from the L2 convergence rates of the nuisance estimators. The non-parametric literature gives rates for kernel regression and sieves/series (e.g. \cite{chen2007large}). For $L_{1}$-penalty estimators such as Lasso see \cite{belloni2011} and \cite{belloni2013least}. For low-level conditions for shrinkage and kernel estimators see Appendix B in \cite{sasaki2021estimation}. Rates for $L_{2}$-boosting in low dimensions are in \cite{zhang2005boosting}, and \cite{kueck2023estimation} find rates for $L_{2}$-boosting with high dimensional data. For results on versions of random forests see \cite{wager2015adaptive} and \cite{athey2019generalized}. For single-layer, sigmoid-based neural networks see \cite{chen1999improved} and for a modern setting of deep neural networks see \cite{farrell2021deep}. Note that $\hat{\varphi}_l$ estimates conditional expectations where both the dependent variables and the conditioning ones are indexed by both $i$ and $j$. \cite{stute1991conditional} calls such objects conditional U-statistics and studies the asymptotic properties of Nadaraya-Watson nonparametric estimators. I run the machine learning algorithms on the stacked pairs. Unfortunately, not much is known about rates for such machine learning regressions which are computationally demanding. Define now the following interaction terms for $\omega \in \{\gamma,\varphi\}$ and let $||\cdot||$ denote the L2 norm.
\begin{align*}
    \hat{\xi}_{ij,ab,l} &= \hat{\varphi}_l(a,X_i,b,X_j,\hat{\gamma}_l) - \varphi(a,X_i,b,X_j,\hat{\gamma}_l) - \hat{\varphi}_l(a,X_i,b,X_j,\gamma) + \varphi(a,X_i,b,X_j,\gamma), \\
    \hat{\xi}_{ij,ab,l}^\omega &= \phi_{ab}(z_{i},z_{j},\hat{\omega}_{l},\hat{\alpha}%
    _{l}^\omega)-\phi_{ab}(z_{i},z_{j},\omega,\hat{\alpha}_{l}^\omega)-\phi_{ab}(z_{i},z_{j},\hat{\omega}_{l},\alpha^\omega)+\phi_{ab}(z_{i},z_{j},\omega,\alpha^\omega).
\end{align*}

\begin{assumption}
    \label{ass_glob_rob_quadr}
    For each $l =1,...,L$
    \begin{enumerate}[(i)]
        \item $\int \int \phi_{ab}^{\gamma}(z_i,z_j,\gamma,\hat{\alpha}_l^{\gamma}) F(dz_i)F(dz_j) = 0$ and $\int \int \phi_{ab}^{\varphi}(Z_i,Z_j,\varphi,\hat{\alpha}_l^{\varphi}) F(dz_i)F(dz_j) = 0$.
        \item $\Ex(||\hat{\gamma}_l - \gamma||^2) = o(n^{-2\lambda_\gamma})$, $\Ex(||\hat{\varphi}_l - \varphi||^2) = o(n^{-2\lambda_\varphi})$ and
            \begin{align*}
                |\Ex[(\varphi(a,X_i,b,X_j,\tilde{\gamma}) + \phi_{ab}^{\gamma}(Z_i,Z_j,\tilde{\gamma},\alpha^{\gamma}))\pi_{ab}(X_i,X_j)]| &\leq C ||\tilde{\gamma} - \gamma||^2 \\
                |\Ex[(\tilde{\varphi}(a,X_i,b,X_j,\gamma) + \phi_{ab}^{\varphi}(Z_i,Z_j,\tilde{\varphi},\alpha^{\varphi}))\pi_{ab}(X_i,X_j)]| &\leq C ||\tilde{\varphi} - \varphi||^2.  
            \end{align*}
    \end{enumerate}
\end{assumption}
Assumption \ref{ass_glob_rob_quadr} (i) is usually verified from visual inspection and (ii) requires L2 convergence rates and some smoothness. $C$ is a constant so the right-hand-sides above do not depend on $\pi \in \Pi$.

\begin{assumption}
\label{ass_interaction_term}
For each $l = 1,...,L$: $\sqrt{n}\Ex(\hat{\xi}_{ij,ab,l}^\omega\pi_{ab}(X_i,X_j)) = o(1).$
\end{assumption}
These are rate conditions on the remainder terms $\hat{\xi}_{l}^\omega(w_{i},w_{j})$.
Often, Assumption \ref{ass_interaction_term} follows if $\sqrt{n}
||\hat{\alpha}_{l}^{\omega}-\alpha||||\hat{\omega}_{l}-\omega|| = o(1)$. In essence, it is enough for the product of the nonparametric estimators to go to zero at a $\sqrt{n}$ rate.

\subsection{Conditions on the complexity of the policy class}

The complexity of the policy class must also be restricted. If all sorts of subsets of $\mathcal{X}$ are allowed to decide who should be treated then we get overfitted policy rules. As in \cite{athey2021policy} I measure the policy class complexity with its VC dimension (see for instance \cite{wainwright2019high}) which is allowed to grow with the sample size. 
\begin{assumption}
    \label{ass_VC}
    There are constants $0< \beta < 1/2$ and $n^*\geq 1$ such that for all $n \geq n^*$, $VC(\Pi_n) < n^\beta$.
\end{assumption}
Examples of finite VC-dimension classes are linear eligibility scores or generalized eligibility scores. Policy classes whose VC-dimension can increase with the sample size are for example decision trees which get deeper with sample size. 

\subsection{Upper bounds}

Let now
\begin{align*}
W(\pi) &= \Ex \biggl[ \sum_{(a,b) \in \{0,1\}^2}\psi_{ab}(Z_i,Z_j,\gamma,\varphi,\alpha)\pi_{ab}(X_i,X_j)\biggr], \\
\T{W}_n(\pi) &= \binom{n}{2}^{-1} \sum_{i<j} \biggl[ \sum_{(a,b) \in \{0,1\}^2}\psi_{ab}(Z_i,Z_j,\gamma,\varphi,\alpha)\pi_{ab}(X_i,X_j)\biggr], \\
\hat{W}_n(\pi) &= \binom{n}{2}^{-1} \sum_{l=1}^L \sum_{(i,j)\in I_l} \biggl[\sum_{(a,b) \in \pi}\psi_{ab}(Z_i,Z_j,\hat{\gamma}_l,\hat{\varphi}_l,\hat{\alpha}_l)\pi_{ab}(X_i,X_j).\biggr],
\end{align*}
$W(\pi)$ and $\T{W}_n(\pi)$ are the welfare and the infeasible estimator of the welfare at policy rule $\pi$ when all nuisance parameters are known. $\hat{W}_n(\pi)$ is the feasible estimator which we already introduced in (\ref{eq_welfare_estimator}). Let $W_{\Pi_n}^* = \sup_{\pi \in \Pi_n} W(\pi)$ be the best possible welfare. I give upper bounds to the regret: $\Ex[W_{\Pi_n}^* - W(\hat{\pi})]$. I start bounding the regret by
\begin{align}
\Ex[W_{\Pi_n}^* - W(\hat{\pi})] &\leq 2\Ex\biggl[\sup_{\pi \in \Pi_n} |\hat{W}_n(\pi) - \T{W}_n(\pi)|\biggr] + 2\Ex\biggl[\sup_{\pi \in \Pi_n} |\T{W}_n(\pi) - W(\pi)|\biggr],
\label{eq_simpleboundregret}
\end{align}
The second term above is a standard centered U-process indexed by $\pi \in \Pi_n$. I start as in \cite{athey2021policy} by showing the rate of convergence of this second term. I work for some fixed $(a,b) \in \{0,1\}^2$ and let $\Pi_{ab,n} = \{\pi_{ab}: \pi \in \Pi_n\}$. The first step is to bound it by the Rademacher complexity which I define as
\[
\mathcal{R}_n(\Pi_{ab,n}) = \Ex_\varepsilon \biggl(\sup_{\pi \in \Pi_n} \biggl| \floor{n/2}^{-1} \sum_{i=1}^{\floor{n/2}}\varepsilon_i  \Gamma_{i,\floor{n/2} + i}^{ab} \pi_{ab}(X_{i},X_{\floor{n/2} + i})\biggr| \biggr),
\]
where $F_\varepsilon$ is the distribution of Rademacher random variables taking value $1$ and $-1$ with equal probability. $(\varepsilon_1,...,\varepsilon_{\floor{n/2}})$ are independent draws from $F_\varepsilon$. The next result gives a bound for the Rademacher complexity which relies on a characterization of U-statistics as dependent sums of independent sums introduced in \cite{hoeffding1963} 
\begin{lemma}
    \label{lemma_rademacher}
    \[
    \Ex\biggl[\sup_{\pi \in \Pi_n} |\T{W}_n(\pi) - W(\pi)|\biggr] \leq \Ex[2 \mathcal{R}_n(\Pi_{ab,n})].
    \]
\end{lemma}
We want an asymptotic upper bound for $\Ex[\mathcal{R}_n(\Pi_{ab})]$. While \cite{kitagawa2018should} and others provide bounds in terms of the maximum of the (bounded) scores, \cite{athey2021policy} provide bounds based on the variance. The next result provides a bound on the Rademacher complexity based on $S_{ab} =  \Ex[ \Gamma_{i,j}^{2 \, ab}]$.
\begin{lemma}
    \label{lemma_boundforrademacher}
    Under Assumptions \ref{ass_nuisance_rates} and \ref{ass_interaction_term}, if $\Gamma_{ij}^{ab}$ has bounded support, then 
    \[
        \Ex[\mathcal{R}_n(\Pi_{ab,n})] = \mathcal{O}\biggl(\sqrt{\frac{S_{ab} \cdot VC(\Pi_{ab,n})}{\floor{n/2}}} \biggr).
    \]    
\end{lemma} 
Lemma \ref{lemma_boundforrademacher} can be generalized to sub-Gaussianity. However, it comes at the cost of making the (already involved) proofs substantially less tractable. Now I provide asymptotic upper bounds for the first term in (\ref{eq_simpleboundregret}). \cite{escanciano2023machine} show that for given $\pi \in \Pi_n$, $\sqrt{n}(\hat{W}_n(\pi) - \T{W}_n(\pi)) \to_p 0$. The next result makes this uniform in $\pi \in \Pi_n$.

\begin{lemma}[Uniform coupling]
\label{lemma_uniform_coupling}
Under Assumptions \ref{ass_nuisance_rates} and \ref{ass_interaction_term}
\[
    \sqrt{n}\Ex[\sup_{\pi \in \Pi_n}|\hat{W}_n(\pi) - \T{W}_n(\pi)|] = \mathcal{O}\biggl(1 + \frac{VC(\Pi_{ab,n})}{\floor{n/2}^{\min(\lambda_\gamma,\lambda_\varphi,\lambda_\alpha)}}\biggr).
\]
\end{lemma}
Finally, using Lemmas \ref{lemma_boundforrademacher} and \ref{lemma_uniform_coupling} the following holds.
\begin{theorem}
    \label{thm_regret_upper_bound}
    If Assumptions \ref{ass_nuisance_rates}, \ref{ass_interaction_term} and \ref{ass_VC} hold with $\beta < \min(\lambda_\gamma,\lambda_\varphi,\lambda_\alpha)$. Then
    \[
        \Ex[W_{\Pi_n}^* - W(\hat{\pi})] = \mathcal{O}\biggl(\sqrt{\frac{S_{ab} \cdot (2VC(\Pi_{n})-1)}{\floor{n/2}}} \biggr).
    \]    
\end{theorem}

\section{Empirical illustration}
\label{sec_empapp}

I study the optimal allocation of children to preschool. I make use of the Panel Study of Income Dynamics (PSID) database. This survey contains a rich set of circumstances and long-term outcomes. In 1995, adults between 18-30 years old were asked about their participation in preschool. The outcome is the average earnings from 25 to 35 years old. I assume selection on sex, birthyear, average parental income in the 5 years before birth, mother's education, father's education, father's occupation and whether the individual is black. In Table \ref{table_ate} we see the results of estimating the Average Treatment Effect (ATE), Gini, IOp and IGM measured by the Kendall-$\tau$.

\begin{table}[h!]
    \centering
    \small
    \begin{tabular}{@{}cccccccc@{}}
      \toprule
      Outcome & ATE & se & p-value & Gini & IOp & IGM & n \\ \midrule
      Earnings 25-35 & 4622.063 & 1083.865 & 0 & 0.392 & 0.172 & 0.168 & 2971 \\ \bottomrule
      \end{tabular}
      \caption{\small ATE, Gini, IOp and Kendal-$\tau$}
    \label{table_ate}
    \end{table}
\noindent To estimate the ATE, I use the doubly robust Augmented Inverse Propensity weighting from \cite{robins1994estimation} using Conditional Inference Forests (CIF) to estimate the regression functions and propensity scores. I chose CIF by cross-validation among a pool of different machine learners. Under the assumption of no selection on observables, we observe a positive effect of attending preschool of 4,622\$ of added annual earnings. Dollars have been adjusted by the CPI to 2010 dollars. The Gini coefficient is $0.39$ and IOp is $0.17$, i.e. almost 44\% of total inequality can be explained by the circumstances we observe. The Kendall-$\tau$ is around 0.17 which indicates a positive association between parental and child incomes.

I compute optimal treatment rules based on parental income and the mother's years of education. I set the target in the Kendall-$\tau$ welfare to zero, meaning that the aim is to completely erase intergenerational persistence. As the policy class, I use 2-depth decision trees. The U-statistic nature of the SWF prevents me from using the computational shortcuts in \cite{athey2021policy} since the sub-trees are not independent optimization problems. Instead, I use the deciles of parental income as cutting points instead of all the observed values of parental income. 

Optimal treatment allocation is the same for additive and IOp welfare. Although this seems surprising, it is possible if decreasing inequality of opportunity leads to sizeable reductions of the average. In fact, as reported in Table \ref{table_welfares}, the estimated rule maximizing the average already drastically decreases IOp. I show the optimal rule under these two welfares in Figure \ref{figure_util}. At the terminal nodes, I report the number of observations, the conditional average treatment effect (CATE) in the node and the proportion of observations node that are treated in the data ($\hat{p})$. For additive/IOp welfare, the first cutting point is whether parental income is below or above the 40th percentile (51,515\$). If an observation is below this cut-off the tree splits according to the education of the mother. If parental income is below the 40th percentile and the mother's education is less than college (below 13 years) the tree allocates the observation to treatment. The CATE in this node is positive so, as we would expect, an additive policy maker treats these observations. If parental income is below the 40th percentile but the mother is highly educated the CATE is negative and hence the additive policy maker does not allocate the individual to treatment. For high parental income, we split next on the mother's education but at a higher level. If your parental income is higher than the 40th percentile and your mother attended college or less (16 years of education or less) you are allocated to treatment and in this node, we have large positive effects. However, we do not allocate kids with high parental income and high maternal education to preschool since the CATE in this group is negative.

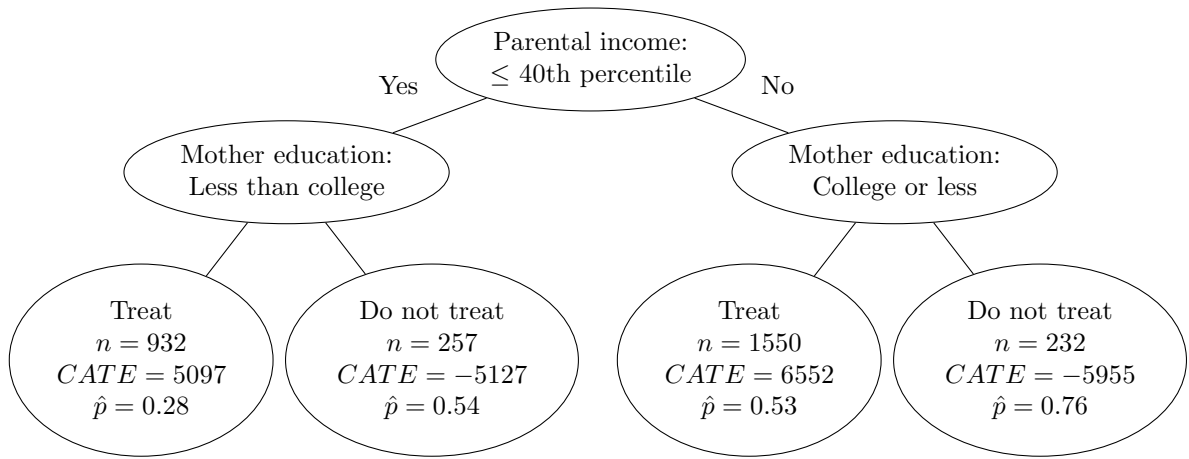
\begin{figure}[h!]
    \centering
    \footnotesize
    \begin{tikzpicture}[
        every node/.style = {minimum width = 2em, draw, ellipse, align=center},
        level 1/.style = {sibling distance = 23em},
        level 2/.style = {sibling distance = 11em, level distance = 25mm}
    ]
    \node {Parental income: \\ $\leq$ 40th percentile}
        child {node {Mother education: \\ Less than college}
            child {node {Treat \\ $n = 932$
                         \\$CATE = 5097$ \\ $\hat{p} = 0.28$}}
            child {node {Do not treat \\ $n = 257$
            \\$CATE = -5127$ \\ $\hat{p} = 0.54$}}
            edge from parent node[left, xshift= 0mm, yshift = 4mm,
            draw = none] {Yes}
        }
        child {node {Mother education: \\ College or less}
            child {node {Treat \\ $n = 1550$
            \\$CATE = 6552$ \\ $\hat{p} = 0.53$}}
            child {node {Do not treat \\ $n = 232$
            \\$CATE = -5955$ \\ $\hat{p} = 0.76$}}
            edge from parent node[right, xshift= 0mm, yshift = 4mm,
            draw = none] {No}
        };
    \end{tikzpicture}
    \caption{\footnotesize Estimated optimal policy rule under additive and IOp welfare.}
    \label{figure_util}
  \end{figure}

\noindent If we take parental income to be a proxy for the quality of preschool, it is enough for the mother to have more than 13 years of education for the child to be better off without preschool. However, for children who attend better preschools (have higher parental income), the mother has to have more than 16 years of education for the child to be better off without preschool. This is in line with results in the psychology and economics literature (see \cite{fort2020cognitive}). To decrease IOp further, we need to treat advantaged kids who do not benefit from preschool. The penalization of IOp is not severe enough to do this. The observed treatment allocation deviates significantly from the optimal rule, this is likely as parents consider more than future earnings when deciding on preschool.

In figure \ref{figure_ineq}, we see the optimal policy rule for the inequality SWF. The tree is the same except for the first cutting point on parental income. We first divide individuals into those with parental income lower and higher than the 20th percentile (37,699\$). Then, the division based on the mother's education is the same. Hence, compared to the previous tree, we shift 20\% of the population to the right side subtree. For instance, a kid who has a parental income of 40,000\$ and whose mother has 16 years of education is not treated under the additive/IOp welfare but is treated under the inequality based optimal rule. Although masked by other observations in the node, this 20\% of the population who is switched to treatment has an estimated negative CATE. When we penalize all sorts of inequalities, it starts becoming optimal to decrease the average outcome to decrease inequality.     

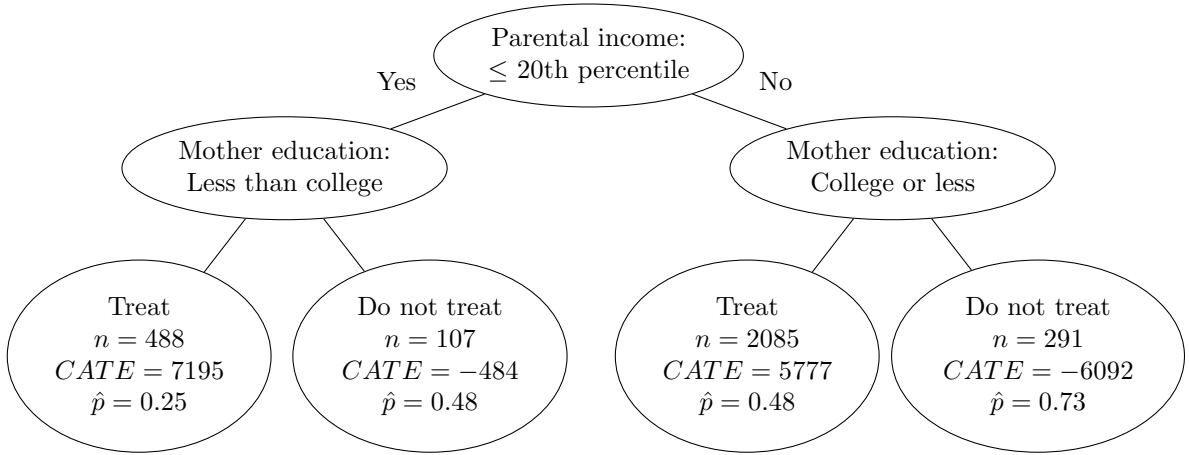
\begin{figure}[h!]
    \centering
    \footnotesize
    \begin{tikzpicture}[
        every node/.style = {minimum width = 2em, draw, ellipse, align=center},
        level 1/.style = {sibling distance = 23em},
        level 2/.style = {sibling distance = 11em, level distance = 25mm}
    ]
    \node {Parental income: \\ $\leq$ 20th percentile}
        child {node {Mother education: \\ Less than college}
            child {node {Treat \\ $n = 488$
                         \\$CATE = 7195$ \\ $\hat{p} = 0.25$}}
            child {node {Do not treat \\ $n = 107$
            \\$CATE = -484$ \\ $\hat{p} = 0.48$}}
            edge from parent node[left, xshift= 0mm, yshift = 4mm,
            draw = none] {Yes}
        }
        child {node {Mother education: \\ College or less}
            child {node {Treat \\ $n = 2085$
            \\$CATE = 5777$ \\ $\hat{p} = 0.48$}}
            child {node {Do not treat \\ $n = 291$
            \\$CATE = -6092$ \\ $\hat{p} = 0.73$}}
            edge from parent node[right, xshift= 0mm, yshift = 4mm,
            draw = none] {No}
        };
    \end{tikzpicture}
    \caption{\footnotesize Estimated optimal policy rule under inequality welfare.}
    \label{figure_ineq}
  \end{figure}

\noindent In Figure \ref{figure_igm} we see the results for an IGM aware SWF. Notice that in the IGM welfare there is no efficiency motive and we target a zero Kendall-$\tau$. Hence, the optimal policy is even more controversial since individuals with positive treatment effects are not treated and individuals with negative treatment effects are treated.  

\begin{figure}[h!]
    \centering
    \footnotesize
    \begin{tikzpicture}[
        every node/.style = {minimum width = 2em, draw, ellipse, align=center},
        level 1/.style = {sibling distance = 23em},
        level 2/.style = {sibling distance = 11em, level distance = 25mm}
    ]
    \node {Mother education: \\ Less than college}
        child {node {Parental income: \\ $\leq$ 50th percentile}
            child {node {Treat \\ $n = 1149$
                         \\$CATE = 6052$ \\ $\hat{p} = 0.29$}}
            child {node {Do not treat \\ $n = 913$
            \\$CATE = 5771$ \\ $\hat{p} = 0.49$}}
            edge from parent node[left, xshift= 0mm, yshift = 4mm,
            draw = none] {Yes}
        }
        child {node {Mother education: \\ College or less}
            child {node {Do not treat \\ $n = 580$
            \\$CATE = 2891$ \\ $\hat{p} = 0.64$}}
            child {node {Treat \\ $n = 329$
            \\$CATE = -5145$ \\ $\hat{p} = 0.71$}}
            edge from parent node[right, xshift= 0mm, yshift = 4mm,
            draw = none] {No}
        };
    \end{tikzpicture}
    \caption{\footnotesize Estimated optimal policy rule under IGM welfare.}
    \label{figure_igm}
  \end{figure}
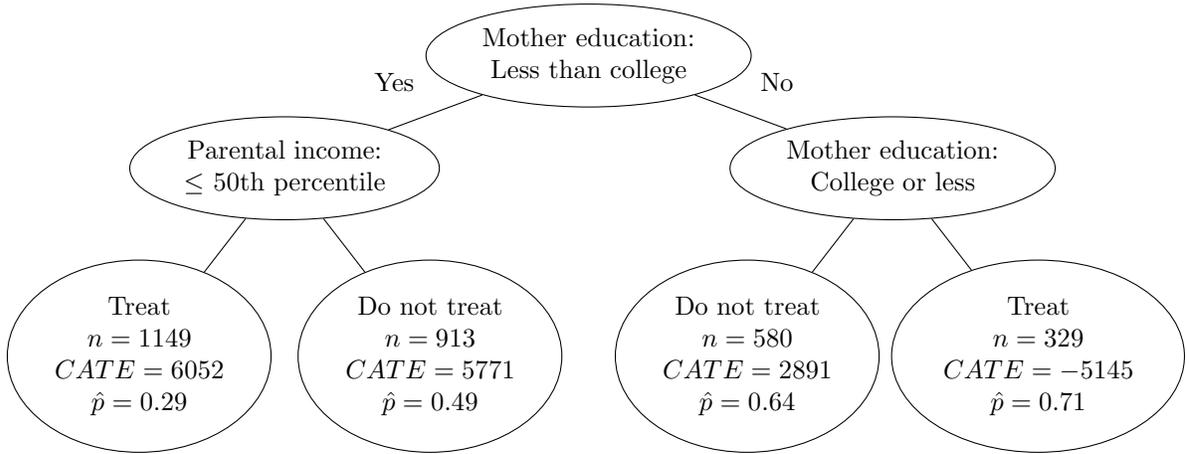

\noindent Finally, in Table \ref{table_welfares} we see a summary of the results and compare the estimated optimal treatments with situations in which no one or everyone is treated. For additive, IOp and inequality welfares, treating no one gives the worst welfare. In the IGM case, treating no one and treating everyone give basically the same welfare (note that maximal welfare in the IGM case is 0). The additive and IOp welfares have an estimated optimal policy rule that attains the highest average outcome and the lowest IOp. The decrease in IOp under this rule is drastic. While in the sample we can explain 44\% of total inequality with circumstances (IOp/Gini), at the optimal additive/IOp rule we explain 35\%. This explains why both rules coincide. In other settings, maximizing the average might increase IOp. The estimated optimal treatment rule for IGM gives a similar Gini and IOp as the the estimated optimal policy rule for IOp or inequality welfare, but at a much larger cost in average outcome. The estimated optimal policy rule under the inequality welfare gives the lowest Gini compared to additive and IOp welfares. Interestingly, it gives the highest IOp across all welfares. IGM rule gives the lowest Kendall-$\tau$. 

    \begin{table}[h!]
        \centering
        \footnotesize
        \begin{tabular}{@{}clcccccc@{}}
            \toprule
             & \multicolumn{1}{c}{} & \textbf{Welfare} & \textbf{Mean} & \textbf{Gini} & \textbf{IOp} & \textbf{IGM} & \textbf{Share treated} \\ \midrule
            \textbf{Additive} & \textbf{Optimal rule} & 39727 & 39727 & 0.392 & 0.138 & 0.15 & 0.84 \\
             & \textbf{Treat no one} & 34169 & 34169 & 0.4 & 0.162 & 0.148 & 0 \\
             & \textbf{Treat everyone} & 38778 & 38778 & 0.383 & 0.142 & 0.15 & 1 \\ \midrule
            \textbf{IOp} & \textbf{Optimal rule} & 34231 & $\cdot$ & $\cdot$ & $\cdot$ & $\cdot$ & $\cdot$ \\
             & \textbf{Treat no one} & 28640 & $\cdot$ & $\cdot$ & $\cdot$ & $\cdot$ & $\cdot$ \\
             & \textbf{Treat everyone} & 33282 & $\cdot$ & $\cdot$ & $\cdot$ & $\cdot$ & $\cdot$ \\ \midrule
            \textbf{Inequality} & \textbf{Optimal rule} & 24165 & 39383 & 0.386 & 0.141 & 0.153 & 0.87 \\
             & \textbf{Treat no one} & 20518 & 34169 & 0.4 & 0.162 & 0.148 & 0 \\
             & \textbf{Treat everyone} & 23942 & 38778 & 0.383 & 0.142 & 0.15 & 1 \\ \midrule
            \textbf{IGM} & \textbf{Optimal rule} & -0.086 & 35951 & 0.383 & 0.139 & 0.086 & 0.5 \\
             & \textbf{Treat no one} & -0.148 & 34169 & 0.4 & 0.162 & 0.148 & 0 \\
             & \textbf{Treat everyone} & -0.15 & 38778 & 0.383 & 0.142 & 0.15 & 1 \\ \midrule
            \textbf{Sample} & $\cdot$ & $\cdot$ & 36197 & 0.392 & 0.172 & 0.168 & 0.47 \\
             \bottomrule
            \end{tabular}
        \caption{\footnotesize Welfare, mean, Gini, IOp, IGM and share treated for different optimal policy rules compared with policies which treat no one and everyone. I also show the actual values observed in the sample. The dots in the IOp rows indicate that the optimal policy rule is the same as in the additive case.}
        \label{table_welfares}
        \end{table}

\noindent Finally, comparing the results with what we observe with the treatment allocation in the sample, we achieve a higher mean with all other welfares except with the IGM welfare. The Gini in the sample is the same as the one under the estimated optimal additive and IOp rule. The observed IOp in the sample is higher than the one achieved under the estimated rules of all other welfares. IGM observed in the sample is the lowest (highest Kendall-$\tau$) compared to all welfares. Finally, the share of treated in the sample is also lower than the one achieved under the estimated optimal rule of all other welfares. However, this could be explained by not taking into account costs of treatment.

\section{Appendix}
\label{Appendix}

\addcontentsline{toc}{section}{Appendices} \renewcommand{\thesubsection}{\Alph{subsection}}

\subsection{Identification and local robustness}
\label{app_ipw}
Here I show identification and local robustness results for (DM) and (IPW).
\begin{proposition}
    \label{app_prop_linid}
    Under Assumption \ref{ass_ident}, $W(\pi)$ is identified as
    \begin{align*}
        W(\pi) &= \Ex[m_1(Z_i,\gamma,\nu)\pi(X_i) + m_0(Z_i,\gamma,\nu)(1-\pi(X_i))],
    \end{align*}
    with $\nu \in \{\varphi,e\}$ and where $m_1$ and $m_0$ can be any of the following
    \begin{align*}
        \text{(DM)} \quad m_1(Z_i,\gamma, \varphi) = 
        \varphi(1,X_i,\gamma_1), \quad
        &m_0(Z_i,\gamma, \varphi) = 
        \varphi(0,X_i,\gamma_0) \\
        \text{(IPW)} \quad m_1(Z_i,\gamma,e) = 
        \frac{g(Y_i,X_i,\gamma_1)D_i}{e(X_i)}, \quad
        &m_0(Z_i,\gamma,e) = 
        \frac{g(Y_i,X_i,\gamma_0)(1-D_i)}{1-e(X_i)},
    \end{align*}
\end{proposition}
Now I show the identification result for U-statistics estimable quantities.

\begin{proposition}
    \label{app_prop_Uid}
    Under Assumption \ref{ass_ident}, $W(\pi)$ in (\ref{eq_Uwf}) is identified in the following ways
    \begin{align*}
        W(\pi) &= \Ex\biggl[\sum_{(a,b)\in \{0,1\}^2} m_{ab}(Z_i,Z_j,\gamma,\nu)\pi_{ab}(X_i,X_j)\biggr],
    \end{align*}
    where $\nu \in \{\varphi, e\}$ and $m_{ab}$ can be any of the following
    \begin{align*}
        \text{(DM)} \quad m_{ab}(Z_i,Z_j,\gamma,\varphi) &= 
        \varphi(a,X_i,b,X_j,\gamma_a,\gamma_b), \\
         \text{(IPW)} \quad m_{ab}(Z_i,Z_j,\gamma,e) &=
        \frac{g(Y_i,X_i,Y_j,X_j,\gamma_a,\gamma_b)D_{ij}^{ab}}{e_{ab}(X_i,X_j)}.
    \end{align*}
\end{proposition}
Next, I introduce the Assumption and proposition to compute locally robust scores.

\begin{assumption}
    \label{app_ass_Ulinearization}
    $\exists$ $\alpha_{ab,p}^\gamma$, $P<\infty$, and $(c_{1p}, c_{2p})$ for $p = 1,...,P$, s.t. for all $(a,b) \in \{0,1\}^2$ 
    \[
    \frac{d}{d\tau} \mathbb{E}[m_{ab}(Z_i,Z_j,\bar{\gamma}_\tau,\nu)] = \mathbb{E}\biggl[\sum_{p=1}^P\alpha_{ab,p}^\gamma(D_i,X_i,D_j,X_j)(c_{1p}\bar{\gamma}_\tau(D_i,X_i) + c_{2p}\bar{\gamma}_\tau(D_j,X_j))\biggr],
    \]
    where $\bar{\gamma}_\tau$ is defined as in Assumption \ref{ass_linlinearization} and $\Ex[\alpha_{ab,p}^\gamma(D_i,X_i,D_j,X_j)^2] < \infty$.
\end{assumption}

\begin{proposition}
    \label{app_prop_Uorthscores}
    The orthogonal scores are given by $\Gamma_{ij}(\pi) = \sum_{(a,b) \in \{0,1\}^2} \Gamma_{ij}^{ab}\pi_{ab}(X_i,X_j)$, where depending on whether we identify with DM or IPW we have
    \begin{align*}
        (DM) \quad \Gamma_{ij}^{ab} &= 
        \varphi(a,X_i,b,X_j,\gamma_a,\gamma_b) + \phi_{ab}^\varphi(D_i,X_i,D_j,X_j,\varphi,\alpha^{e}) + \phi_{ab}^\gamma(D_i,X_i,D_j,X_j,\gamma,\alpha^{\gamma}) \\
        (IPW) \quad \Gamma_{ij}^{ab} &=
        \frac{g(Y_i,X_i,Y_j,X_j,\gamma_a,\gamma_b)D_{ij}^{ab}}{e_{ab}(X_i,X_j)} + \phi_{ab}^e(D_i,X_i,D_j,X_j,e,\alpha^{e}) + \phi_{ab}^\gamma(D_i,X_i,D_j,X_j,\gamma,\alpha^{\gamma}),
    \end{align*}
    where
    \begin{align*}
        \phi_{ab}^\gamma(D_i,X_i,D_j,X_j,e,\alpha^{\gamma}) &= \sum_{p=1}^P \alpha_{ab,p}^\gamma(D_i,X_i,D_j,X_j,e)(c_{1p} Y_i + c_{2p} Y_j - c_{1p}\gamma(D_i,X_i) - c_{2p} \gamma(D_j,X_j)), \\
        \phi_{ab}^\varphi(D_i,X_i,D_j,X_j,\varphi,\alpha^{m}) &=  \alpha_{ab}^\varphi(D_i,X_i,D_j,X_j)(g(Y_i,X_i,Y_j,X_j,\gamma_a,\gamma_b) - \varphi(D_i,X_i,D_j,X_j,\gamma_a,\gamma_b)),\\
        \phi_{ab}^e(D_i,X_i,D_j,X_j,e,\alpha^{e}) &= \alpha_{ab,1}^e(X_i)(\ind(D_i = a) - e_a(X_i)) + \alpha_{ab,2}^e(X_j)(\ind(D_j = b) - e_b(X_j)), \\
        \alpha_{ab}^\varphi(D_i,X_i,D_j,X_j) &= \frac{D_{ij}^{ab}}{e_{ab}(X_i,X_j)}, \, \alpha_{ab,1}^e(X_i) = -\Ex\biggl[\frac{g(Y_i,X_i,Y_j,X_j,\gamma_a,\gamma_b)D_{ij}^{ab}}{e_a(X_i)^2 e_b(X_j)} \biggr| X_i\biggr], \\
    \alpha_{ab,2}^e(X_j) &= -\Ex\biggl[\frac{g(Y_i,X_i,Y_j,X_j,\gamma_a,\gamma_b)D_{ij}^{ab}}{e_a(X_i) e_b(X_j)^2}\biggr|X_j\biggr].
    \end{align*}
\end{proposition}

\subsection{Proofs of main results}

\begin{proof}[Proof of Propositions \ref{prop_linid}, \ref{prop_Uid}, \ref{prop_Uorthscores}] 
    See Proofs of Propositions \ref{app_prop_linid}, \ref{app_prop_Uid}, \ref{app_prop_Uorthscores}.
\end{proof}

\begin{proof}[Proof of Proposition \ref{app_prop_linid}]
    I proof only the identification of the first term of the welfare since the second one follows in the same manner.
    \begin{align*}
        \mathbb{E}[g(Y_i(1),X_i,\gamma^{(1)}) \pi(X_i)] &= 
        \mathbb{E}[\Ex(g(Y_i(1),X_i,\gamma_1)|X_i) \pi(X_i)] \\
        &= \mathbb{E}[\Ex(g(Y_i,X_i,\gamma_1)|D_i = 1, X_i) \pi(X_i)] \\
        &= \mathbb{E}\biggl[\Ex \biggl(\frac{g(Y_i,X_i,\gamma_1)D_i}{e(X_i)}|X_i\biggr) \pi(X_i)\biggr] \\
        &= \mathbb{E}\biggl[\frac{g(Y_i,X_i,\gamma_1)D_i}{e(X_i)} \pi(X_i)\biggr],
    \end{align*}
    the third equality already establishes the identification by the direct method.
\end{proof}

\begin{proof}[Proof of Proposition \ref{prop_linorthscores}]
    Let $d/d\tau$ be the derivative with respect to $\tau$ evaluated at $\tau = 0$, let $\varphi_\tau = \varphi + \tau \tilde{\varphi}$ for some $\tilde{\varphi}$ in the space where $\varphi$ lives and $\Ex_\tau$ be the expectation with respect to $F + \tau(H-F)$ for some alternative distribution $H$. Then
    \begin{align*}
        \frac{d}{d\tau} \Ex[\varphi_\tau(1,X_i,\bar{\gamma}_\tau(1,X_i))\pi(X_i)] &= \frac{d}{d\tau} \Ex[\varphi_\tau(1,X_i,\gamma(1,X_i))\pi(X_i)] + \frac{d}{d\tau} \Ex[\varphi(1,X_i,\bar{\gamma}_\tau(1,X_i))\pi(X_i)].
    \end{align*}
    For the first term note that 
    \begin{align*}
        \frac{d}{d\tau} \Ex[\varphi_\tau(1,X_i,\gamma(1,X_i))\pi(X_i)] &= \frac{d}{d\tau} \Ex\biggl[\frac{D_i}{e(X_i)}\varphi_\tau(1,X_i,\gamma(1,X_i))\pi(X_i)\biggr] \\
        &= \frac{d}{d\tau} \Ex_\tau\biggl[\frac{D_i}{e(X_i)}\varphi_\tau(1,X_i,\gamma(1,X_i))\pi(X_i)\biggr] \\
        &- \frac{d}{d\tau} \Ex_\tau\biggl[\frac{D_i}{e(X_i)}\varphi(1,X_i,\gamma(1,X_i))\pi(X_i)\biggr] \\
        &= \frac{d}{d\tau} \Ex_\tau\biggl[\frac{D_i}{e(X_i)}(g(Y_i,X_i,\gamma(1,X_i)) - \varphi(1,X_i,\gamma(1,X_i)))\pi(X_i)\biggr],
    \end{align*}
    using LIE in the first equality, then the chain rule and that $\varphi_\tau(1,X_i,\gamma(1,X_i))$ is a projection of $g(Y_i, X_i,\gamma(1,X_i))$. For the second term
    \begin{align*}
        \frac{d}{d\tau} \Ex[\varphi(1,X_i,\bar{\gamma}_\tau(1,X_i))\pi(X_i)] &= \frac{d}{d\tau} \Ex[\alpha_1(D_i,X_i)\bar{\gamma}_\tau(1,X_i)\pi(X_i)] \\
        &= \frac{d}{d\tau} \Ex_\tau[\alpha_1(D_i,X_i)\bar{\gamma}_\tau(1,X_i)\pi(X_i)] \\
        &- \frac{d}{d\tau} \Ex_\tau[\alpha_1(D_i,X_i)\gamma_\tau(1,X_i)\pi(X_i)] \\
        &= \frac{d}{d\tau} \Ex_\tau[\alpha_1(D_i,X_i)(Y_i - \gamma(1,X_i))\pi(X_i)],
    \end{align*} 
   using Assumption \ref{ass_linlinearization} in the first equality, then the chain rule and then the fact that $\gamma$ is a projection. Then, following \cite{chernozhukov2022locally} we have that 
    \[
    \Gamma_{1i} = \varphi(1,X_i,\gamma) + \frac{D_i}{e(X_i)}(g(Y_i,X_i,\gamma_1) - \varphi(1,X_i,\gamma)) + \alpha_1(D_i,X_i)(Y_i-\gamma(D_i,X_i)).   
    \]
    The arguments for $\Gamma_{0i}$ are the analogous.
\end{proof}

\begin{proof}[Proof of Proposition \ref{app_prop_Uid}]
    \begin{align*}
        W(\pi) &= \mathbb{E}\biggl[ \sum_{(a,b) \in \{0,1\}^2}g(Y_i(a),X_i,Y_j(b),X_j,\gamma^{(a)},\gamma^{(b)}) \pi_{ab}(X_i,X_j) \biggr] \\ 
        &= \mathbb{E}\biggl[ \Ex \biggl( \sum_{(a,b) \in \{0,1\}^2}g(Y_i(a),X_i,Y_j(b),X_j,\gamma^{(a)},\gamma^{(b)}) \biggr| X_i, X_j \biggr) \pi_{ab}(X_i,X_j) \biggr] \\
        &= \mathbb{E}\biggl[ \Ex \biggl( \sum_{(a,b) \in \{0,1\}^2}g(Y_i(a),X_i,Y_j(b),X_j,\gamma^{(a)},\gamma^{(b)}) \biggr| X_i, D_i = a, X_j, D_j = b \biggr) \pi_{ab}(X_i,X_j) \biggr] \\
        &= \mathbb{E}\biggl[ \Ex \biggl( \sum_{(a,b) \in \{0,1\}^2}g(Y_i,X_i,Y_j,X_j,\gamma^{(a)},\gamma^{(b)}) \biggr| X_i, D_i = a, X_j, D_j = b \biggr) \pi_{ab}(X_i,X_j) \biggr] \\
        &= \mathbb{E}\biggl[ \Ex \biggl( \sum_{(a,b) \in \{0,1\}^2} \frac{g(Y_i,X_i,Y_j,X_j,\gamma^{(a)},\gamma^{(b)}) D_{ij}^{ab}}{e_{ab}(X_i,X_j)} \biggr| X_i, X_j\biggr) \pi_{ab}(X_i,X_j) \biggr] \\
        &= \mathbb{E}\biggl[ \sum_{(a,b) \in \{0,1\}^2} \frac{g(Y_i,X_i,Y_j,X_j,\gamma^{(a)},\gamma^{(b)}) D_{ij}^{ab}}{e_{ab}(X_i,X_j)}  \pi_{ab}(X_i,X_j) \biggr],
    \end{align*}
    where in the second equality I use LIE, in the third I use selection on observables, in the fourth I use the definition of $Y_i$. The identification by the direct method is in the fourth equality while the IPW is the last equality.  
\end{proof}

\begin{proof}[Proof of Proposition \ref{app_prop_Uorthscores}]
    Let us start with the DM identification. As usual, let $d/d\tau$ be the derivative at $\tau = 0$. Let me also make the dependence on $\varphi$ explicit: $m_{ab}(Z_i,Z_j,\gamma,\varphi) = \varphi(a,X_i,b,X_j,\gamma_a,\gamma_b)$, let also $\varphi_\tau = \varphi + \tau \tilde{\varphi}$ for some $\tilde{\varphi} \in L_2$. By the chain rule
    \begin{align*}
        \frac{d}{d\tau} \Ex[m_{ab}(Z_i,Z_j,\bar{\gamma}_\tau, \varphi_\tau)] &= 
        \frac{d}{d\tau} \Ex[m_{ab}(Z_i,Z_j,\bar{\gamma}_\tau, \varphi)] + 
        \frac{d}{d\tau} \Ex[m_{ab}(Z_i,Z_j,\gamma, \varphi_\tau)].
    \end{align*}
    By Assumption \ref{app_ass_Ulinearization} we have that the first term is
    \[
        \frac{d}{d\tau} \Ex[m_{ab}(Z_i,Z_j,\bar{\gamma}_\tau, \varphi)] = \mathbb{E}\biggl[\sum_{p=1}^P\alpha_{ab,p}^\gamma(D_i,X_i,D_j,X_j)(c_{1p}\bar{\gamma}_\tau(D_i,X_i) + c_{2p}\bar{\gamma}_\tau(D_j,X_j))\biggr],
    \]
    so by Lemma 1 and equation (2.16) in \cite{escanciano2023machine} we have that 
    \[
        \phi_{ab}^\gamma(D_i,X_i,D_j,X_j,e,\alpha^{\gamma}) = \sum_{p=1}^P \alpha_{ab,p}^\gamma(D_i,X_i,D_j,X_j,e)(c_{1p} Y_i + c_{2p} Y_j - c_{1p}\gamma(D_i,X_i) - c_{2p} \gamma(D_j,X_j)).    
    \]
    For the second term notice that
    \begin{align*}
        \Ex[\varphi(a,X_i,b,X_j,\gamma_a,\gamma_b)] &= \Ex\biggl[\varphi(a,X_i,b,X_j,\gamma_a,\gamma_b)\frac{D_{ij}^{ab}}{e_{ab}(X_i,X_j)}\biggr] \\
        &= \Ex\biggl[\varphi(a,X_i,b,X_j,\gamma_a,\gamma_b)\frac{1}{e_{ab}(X_i,X_j)} \biggl| D_{ij}^{ab} = 1 \biggr] \Pr(D_i = a, D_j = b) \\
        &= \Ex\biggl[\varphi(D_i,X_i,D_j,X_j,\gamma_a,\gamma_b)\frac{D_{ij}^{ab}}{e_{ab}(X_i,X_j)}  \biggr].
    \end{align*}
    So by the same arguments
    \[
        \phi_{ab}^\varphi(D_i,X_i,D_j,X_j,\varphi,\alpha^{m}) =  \alpha_{ab}^\varphi(D_i,X_i,D_j,X_j)(g(Y_i,X_i,Y_j,X_j,\gamma_a,\gamma_b) - \varphi(D_i,X_i,D_j,X_j,\gamma_a,\gamma_b)),  
    \]
    with $\alpha_{ab}^\varphi(D_i,X_i,D_j,X_j) = D_{ij}^{ab}/e_{ab}(X_i,X_j)$. For the IPW identification let me make the dependence on the propensity score explicit: $m_{ab}(Z_i,Z_j,\gamma,\varphi,e) = g(Y_i,X_i,Y_j,X_j,\gamma_a,\gamma_b)D_{ij}^{ab}/e_{ab}(X_i,X_j)$. For $c \in \{0,1\}$ let $e_{c,\tau} = e_c + \tau \tilde{e}_c$ for some $\tilde{e}_c \in L_2$ and $e_\tau = (e_{a,\tau},e_{b,\tau})$. Then
    \begin{align*}
        \frac{d}{d\tau} \Ex[m_{ab}(Z_i,Z_j,\bar{\gamma}_\tau, e_\tau)] &= 
        \frac{d}{d\tau} \Ex[m_{ab}(Z_i,Z_j,\bar{\gamma}_\tau, e)] + 
        \frac{d}{d\tau} \Ex[m_{ab}(Z_i,Z_j,\gamma, e_\tau)].
    \end{align*}
    For the first term, we have the same result by using Assumption \ref{app_ass_Ulinearization}. The second 
    \begin{align*}
        \frac{d}{d\tau} \Ex \biggl[\frac{g(Y_i,X_i,Y_j,X_j,\gamma_a,\gamma_b)D_{ij}^{ab}}{e_{a,\tau}(X_i)e_{b,\tau}(X_j)}\biggr] &= 
        \frac{d}{d\tau} \Ex \biggl[\frac{-g(Y_i,X_i,Y_j,X_j,\gamma_a,\gamma_b)D_{ij}^{ab}}{e_{a}(X_i)^2e_{b}(X_j)}e_{a,\tau}(X_i)\biggr] \\
        &+  \frac{d}{d\tau} \Ex \biggl[\frac{-g(Y_i,X_i,Y_j,X_j,\gamma_a,\gamma_b)D_{ij}^{ab}}{e_{a}(X_i)e_{b}(X_j)^2}e_{b,\tau}(X_j)\biggr] \\
        &= \frac{d}{d\tau} \Ex \biggl[\Ex \biggl(\frac{-g(Y_i,X_i,Y_j,X_j,\gamma_a,\gamma_b)D_{ij}^{ab}}{e_{a}(X_i)^2e_{b}(X_j)} \biggl| X_i \biggr) e_{a,\tau}(X_i)\biggr] \\
        &+  \frac{d}{d\tau} \Ex \biggl[\Ex \biggl(\frac{-g(Y_i,X_i,Y_j,X_j,\gamma_a,\gamma_b)D_{ij}^{ab}}{e_{a}(X_i)e_{b}(X_j)^2} \biggl| X_j \biggr) e_{b,\tau}(X_j)\biggr].
    \end{align*}
    So by the same arguments as before 
    \[
        \phi_{ab}^e(D_i,X_i,D_j,X_j,e,\alpha^{e}) = \alpha_{ab,1}^e(X_i)(\ind(D_i = a) - e_a(X_i)) + \alpha_{ab,2}^e(X_j)(\ind(D_j = b) - e_b(X_j)),    
    \]
    where
    \begin{align*}
        \alpha_{ab,1}^e(X_i) &= -\Ex\biggl[\frac{g(Y_i,X_i,Y_j,X_j,\gamma_a,\gamma_b)D_{ij}^{ab}}{e_a(X_i)^2 e_b(X_j)} \biggr| X_i\biggr], \, \alpha_{ab,2}^e(X_j) &= -\Ex\biggl[\frac{g(Y_i,X_i,Y_j,X_j,\gamma_a,\gamma_b)D_{ij}^{ab}}{e_a(X_i) e_b(X_j)^2}\biggr|X_j\biggr].
    \end{align*}
\end{proof}

\begin{proof}[Proof of Proposition \ref{prop_ioporthscores}]
    Let $\gamma_{c,\tau} = \gamma_c + \tau \tilde{\gamma}_c$ for some $\tilde{\gamma_c} \in L_2$. For $(a,b) \in \{0,1\}^2$
    \begin{align*}
        \frac{d}{d\tau} \Ex[\gamma_{a,\tau}(X_i) + \gamma_{b,\tau}(X_j)] &= 
        \frac{d}{d\tau} \Ex\biggl[ \gamma_{a,\tau}(X_i)\frac{\ind(D_i = a)}{e_a(X_i)} + \gamma_{b,\tau}(X_j)\frac{\ind(D_j = b)}{e_b(X_j)}\biggr] \\
        &= \frac{d}{d\tau} \Ex\biggl[ \bar{\gamma}_{\tau}(D_i,X_i)\frac{\ind(D_i = a)}{e_a(X_i)} + \bar{\gamma}_{\tau}(D_j,X_j)\frac{\ind(D_j = b)}{e_b(X_j)}\biggr] \\
        &= \frac{d}{d\tau} \Ex_\tau \biggl[ \bar{\gamma}_{\tau}(D_i,X_i)\frac{\ind(D_i = a)}{e_a(X_i)} + \bar{\gamma}_{\tau}(D_j,X_j)\frac{\ind(D_j = b)}{e_b(X_j)}\biggr] \\
        &- \frac{d}{d\tau} \Ex_\tau \biggl[ \gamma(D_i,X_i)\frac{\ind(D_i = a)}{e_a(X_i)} + \gamma(D_j,X_j)\frac{\ind(D_j = b)}{e_b(X_j)}\biggr] \\
        &= \frac{d}{d\tau} \Ex_\tau\biggl[ \frac{\ind(D_i = a)}{e_a(X_i)}(Y_i - \gamma(D_i,X_i)) + \frac{\ind(D_j = b)}{e_b(X_j)}(Y_j - \gamma(D_j,X_j))\biggr].
    \end{align*}
    Also, let $\Delta_{a,b} = \gamma_a(X_i) - \gamma_b(X_j)$, then 
    \begin{align*}
        \frac{d}{d\tau} \Ex[|\gamma_{a,\tau}(X_i) - \gamma_{b,\tau}(X_j)|] &= 
        \frac{d}{d \tau} \Ex[|\Delta_{ab} + \tau(\tilde{\gamma}_a(X_i) - \tilde{\gamma}_b(X_j))|].
    \end{align*}
    As shown in \cite{escanciano2023machine}, the Gateaux derivative of the mapping $\Delta \mapsto \Ex(|\Delta|)$ is some direction $\nu$ (assuming no point mass at zero, which follows from the assumptions in the Proposition) is $\Ex[sgn(\Delta)\nu]$. Hence, by the chain rule
    \begin{align*}
        \frac{d}{d\tau} &\Ex[\gamma_{a,\tau}(X_i) + \gamma_{b,\tau}(X_j)] = 
        \frac{d}{d\tau} \Ex[sgn(\gamma_{a}(X_i) - \gamma_{b}(X_j))(\gamma_{a,\tau}(X_i) - \gamma_{b,\tau}(X_j))] \\
        &= \frac{d}{d\tau} \Ex\biggl[sgn(\gamma_{a}(X_i) - \gamma_{b}(X_j))\biggl(\gamma_{a,\tau}(X_i)\frac{\ind(D_i = a)}{e_a(X_i)} - \gamma_{b,\tau}(X_j)\frac{\ind(D_j = b)}{e_b(X_j)}\biggr)\biggr] \\
        &= \frac{d}{d\tau} \Ex\biggl[sgn(\gamma_{a}(X_i) - \gamma_{b}(X_j))\biggl(\gamma_{\tau}(D_i,X_i)\frac{\ind(D_i = a)}{e_a(X_i)} - \gamma_{\tau}(D_j,X_j)\frac{\ind(D_j = b)}{e_b(X_j)}\biggr)\biggr] \\
        &= \frac{d}{d\tau} \Ex_\tau \biggl[sgn(\gamma_{a}(X_i) - \gamma_{b}(X_j))\biggl(\gamma_{\tau}(D_i,X_i)\frac{\ind(D_i = a)}{e_a(X_i)} - \gamma_{\tau}(D_j,X_j)\frac{\ind(D_j = b)}{e_b(X_j)}\biggr)\biggr] \\
        &- \frac{d}{d\tau} \Ex_\tau\biggl[sgn(\gamma_{a}(X_i) - \gamma_{b}(X_j))\biggl(\gamma(D_i,X_i)\frac{\ind(D_i = a)}{e_a(X_i)} - \gamma(D_j,X_j)\frac{\ind(D_j = b)}{e_b(X_j)}\biggr)\biggr] \\
        &= \frac{d}{d\tau} \Ex_\tau\biggl[sgn(\gamma_{a}(X_i) - \gamma_{b}(X_j))\biggl(\frac{\ind(D_i = a)}{e_a(X_i)}(Y_i - \gamma(D_i,X_i)) - \frac{\ind(D_j = b)}{e_b(X_j)}(Y_j - \gamma(D_j,X_j))\biggr)\biggr].
    \end{align*}
    So by the results in \cite{escanciano2023machine}, the locally robust score is given by
    \begin{align*}
        2\Gamma_{ij}^{ab} &= \gamma_a(X_i) + \gamma_b(X_j) - |\gamma_a(X_i) - \gamma_b(X_j)| \\
        &+  \frac{\ind(D_i = a)}{e_a(X_i)}(Y_i - \gamma(D_i,X_i)) + \frac{\ind(D_j = b)}{e_b(X_j)}(Y_j - \gamma(D_j,X_j)) \\
        &-  sgn(\gamma_{a}(X_i) - \gamma_{b}(X_j))\biggl(\frac{\ind(D_i = a)}{e_a(X_i)}(Y_i - \gamma(D_i,X_i)) - \frac{\ind(D_j = b)}{e_b(X_j)}(Y_j - \gamma(D_j,X_j))\biggr) \\
        &= \gamma_a(X_i) + \gamma_b(X_j) - |\gamma_a(X_i) - \gamma_b(X_j)| \\
        &+  (1 - sgn(\gamma_{a}(X_i) - \gamma_{b}(X_j)))\frac{\ind(D_i = a)}{e_a(X_i)}(Y_i - \gamma(D_i,X_i)) \\
        &+ (1 + sgn(\gamma_{a}(X_i) - \gamma_{b}(X_j)))\frac{\ind(D_j = b)}{e_b(X_j)}(Y_j - \gamma(D_j,X_j)).
    \end{align*}
\end{proof}

\bigbreak

\begin{proof}[Proof of Lemma \ref{lemma_rademacher}]
    Using the definition of $W(\pi)$ and $\T{W}_n(\pi)$ and the triangle inequality we know that 
    \begin{align*}
        \Ex\biggl[\sup_{\pi \in \Pi} |\T{W}_n(\pi) - W(\pi)|\biggr] &= \Ex\biggl[\sup_{\pi \in \Pi} \biggl| \U_n \sum_{(a,b) \in \{0,1\}^2} \biggl(\Gamma_{ij}^{ab} \pi_{ab}(X_i,X_j) - \Ex[\Gamma_{ij}^{ab} \pi_{ab}(X_i,X_j)]\biggr) \biggr| \biggr] \\
        &\leq \sum_{(a,b) \in \{0,1\}^2} \Ex\biggl[\sup_{\pi \in \Pi} \biggl| \U_n  \biggl(\Gamma_{ij}^{ab} \pi_{ab}(X_i,X_j) - \Ex[\Gamma_{ij}^{ab} \pi_{ab}(X_i,X_j)]\biggr) \biggr| \biggr].
    \end{align*}
    By the representation in Section A Online Appendix we can rewrite the above as
    \begin{align}
        \sum_{(a,b) \in \{0,1\}^2} &\Ex\biggl[\sup_{\pi \in \Pi} \biggl| \frac{1}{n!} \sum_{\kappa} \floor{n/2}^{-1} \sum_{i=1}^{\floor{n/2}} \biggl(\Gamma_{\kappa(i)\kappa(\floor{n/2} + i)}^{ab} \pi_{ab}(X_{\kappa(i)},X_{\kappa(\floor{n/2} + i)}) \nonumber \\
        \label{eq_perms}
        &- \Ex[\Gamma_{\kappa(i)\kappa(\floor{n/2} + i)}^{ab} \pi_{ab}(X_{\kappa(i)},X_{\kappa(\floor{n/2} + i)})]\biggr) \biggr| \biggr].
    \end{align}
    Introduce an independent ghost sample $(Z'_1,...,Z'_n)$ which is distributed as $(Z_1,...,Z_n)$, Rademacher random variables $\varepsilon_i$, $i = 1,...,n$, such that $\mathbb{P}(\varepsilon_i = 1) = \mathbb{P}(\varepsilon_i = -1) = 1/2$ and construct ghost scores $\Gamma^{'ab}_{ij}$ using the ghost sample. Let $\mathbb{E}_Z$ be the expectation with respect to the distribution of the sample $(Z_1,...,Z_n)$ and define $\mathbb{E}_{Z'}$ and $\mathbb{E}_\varepsilon$ similarly. Define the Rademacher complexity as
    \[
    \mathcal{R}_n(\Pi) = \Ex_\varepsilon \biggl(\sup_{\pi \in \Pi} \biggl| \floor{n/2}^{-1} \sum_{i=1}^{\floor{n/2}}\varepsilon_i  \Gamma_{i,\floor{n/2} + i}^{ab} \pi_{ab}(X_{i},X_{\floor{n/2} + i})\biggr| \biggr).
    \]
    We are now ready to use a classical symmetrization argument, (\ref{eq_perms}) is equal to
    \begin{align*}
        \sum_{(a,b) \in \{0,1\}^2} &\Ex_Z\biggl[\sup_{\pi \in \Pi} \biggl| \frac{1}{n!} \sum_{\kappa} \floor{n/2}^{-1} \sum_{i=1}^{\floor{n/2}} \biggl(\Gamma_{\kappa(i)\kappa(\floor{n/2} + i)}^{ab} \pi_{ab}(X_{\kappa(i)},X_{\kappa(\floor{n/2} + i)})\\
        &- \Ex_{Z'}[\Gamma_{\kappa(i)\kappa(\floor{n/2} + i)}^{'ab} \pi_{ab}(X'_{\kappa(i)},X'_{\kappa(\floor{n/2} + i)})]\biggr) \biggr| \biggr] \\
        &\leq \frac{1}{n!} \sum_{\kappa}  \sum_{(a,b) \in \{0,1\}^2} \Ex_{Z,Z'}\biggl[\sup_{\pi \in \Pi} \biggl| \floor{n/2}^{-1} \sum_{i=1}^{\floor{n/2}} \biggl(\Gamma_{\kappa(i)\kappa(\floor{n/2} + i)}^{ab} \pi_{ab}(X_{\kappa(i)},X_{\kappa(\floor{n/2} + i)})\\
        &- \Gamma_{\kappa(i)\kappa(\floor{n/2} + i)}^{'ab} \pi_{ab}(X'_{\kappa(i)},X'_{\kappa(\floor{n/2} + i)})\biggr) \biggr| \biggr] \\
        &= \sum_{(a,b) \in \{0,1\}^2} \Ex_{Z,Z',\varepsilon}\biggl[\sup_{\pi \in \Pi} \biggl| \floor{n/2}^{-1} \sum_{i=1}^{\floor{n/2}} \varepsilon_i \biggl(\Gamma_{i,\floor{n/2} + i}^{ab} \pi_{ab}(X_{i},X_{\floor{n/2} + i})\\
        &- \Gamma_{i,\floor{n/2} + i}^{'ab} \pi_{ab}(X'_{i},X'_{\floor{n/2} + i})\biggr) \biggr| \biggr] \\
        &\leq \sum_{(a,b) \in \{0,1\}^2} \Ex_{Z,Z',\varepsilon}\biggl[\sup_{\pi \in \Pi} \biggl| \floor{n/2}^{-1} \sum_{i=1}^{\floor{n/2}} \varepsilon_i \Gamma_{i,\floor{n/2} + i}^{ab} \pi_{ab}(X_{i},X_{\floor{n/2} + i}) \biggr| \\
        &+ \biggl| \floor{n/2}^{-1} \sum_{i=1}^{\floor{n/2}} \varepsilon_i \Gamma_{i,\floor{n/2} + i}^{'ab} \pi_{ab}(X'_{i},X'_{\floor{n/2} + i}) \biggr| \biggr] \\
        &= \sum_{(a,b) \in \{0,1\}^2} \Ex[2\mathcal{R}_n(\Pi)].
    \end{align*}
The first inequality follows from Jensen's and triangle inequalities, then we use that $(Z_{\pi(i)},Z_{\pi(\floor{n/2} + i)},Z_{\pi(i)}',Z_{\pi(\floor{n/2} + i)}')$ is identically distributed for $i = 1,...,\floor{n/2}$ for all permutations in $\kappa$ and that $\varepsilon_i (\Gamma_{i,\floor{n/2} + i}^{ab} \pi_{ab}(X_{i},X_{\floor{n/2} + i}) - \Gamma_{i,\floor{n/2} + i}^{'ab} \pi_{ab}(X'_{i},X'_{\floor{n/2} + i}))$ and $\Gamma_{i,\floor{n/2} + i}^{ab} \pi_{ab}(X_{i},X_{\floor{n/2} + i}) - \Gamma_{i,\floor{n/2} + i}^{'ab} \pi_{ab}(X'_{i},X'_{\floor{n/2} + i})$ have the same distribution, the third inequality uses triangle inequality and the last equality uses $Z_i \sim Z'_i$ and the definition of the Rademacher complexity.
\end{proof}

\begin{proof}[Proof of Lemma \ref{lemma_boundforrademacher}]
    Note that Lemma 2 Online Appendix gives a sequence of covers $B_k$ for $k = 0,...,K$ of $\tilde{\Pi}_{ab}$ of radius less than $2^{-k}$ for some $K$. For any $j = 1,...,J$ with $J = \ceil{\log_2(\floor{n/2})(1-\beta)}$ and $\pi \in \tilde{\Pi}_{ab}$ let $b_j: \tilde{\Pi}_{ab} \mapsto \tilde{\Pi}_{ab}$ be such that $b_j(\pi)$ is an approximating policy from the cover $B_j$ such that $D_n(\pi,b_j(\pi)) \leq 2^{-j}$, such an approximation exists by the same Lemma. Also by the same Lemma, $|\{b_j(\pi): \pi \in \tilde{\Pi}_{ab}\}| \leq  N_{D_n}(2^{-j},\T{\Pi}_{ab},\{X_i,\Gamma_{i,\floor{n/2}+i}\}_{i=1}^{\floor{n/2}})$. Let $\underline{J} = \ceil{1/2\log_2(\floor{n/2})(1-\beta)}$. By a telescope sum and approximations $b_0,...,b_J$,  the complexity decomposes as 
    \begin{align*}
        \mathcal{R}_n(\Pi) &= \Ex_\varepsilon \biggl\{\sup_{\pi \in \Pi} \biggl| \floor{n/2}^{-1} \sum_{i=1}^{\floor{n/2}}\varepsilon_i  \Gamma_{i,\floor{n/2} + i}^{ab}\biggl[b_0(\pi_{ab}(X_{i},X_{\floor{n/2} + i})) \\
        &+ \sum_{j=1}^{\underline{J}} \biggl(b_j(\pi_{ab}(X_{i},X_{\floor{n/2} + i})) - b_{j-1}(\pi_{ab}(X_{i},X_{\floor{n/2} + i}))\biggr) \\
        &+ (b_J(\pi_{ab}(X_{i},X_{\floor{n/2} + i})) - b_{\underline{J}}(\pi_{ab}(X_{i},X_{\floor{n/2} + i}))) \\
        &+ (\pi_{ab}(X_{i},X_{\floor{n/2} + i}) - b_{J}(\pi_{ab}(X_{i},X_{\floor{n/2} + i})))\biggr] \biggr| \biggr\}.
    \end{align*} 
    Since $D_n$ is bounded by $1$, by the second property in Lemma 2 Online Appendix we have that $b_0$ can be any policy in $\tilde{\Pi}_{ab}$. Hence, we can set $b_0(\pi_{ab}(X_{i},X_{\floor{n/2} + i})) = 0$ for all $i =1,...,\floor{n/2}$. We approach each of the terms above in turn. Note that $b_0,...,b_J$ is a sequence of increasingly accurate approximations. Notice that the last term above is negligible, i.e. the term involving the closest approximation vanishes at a $\sqrt{n}$ rate. By using Cauchy-Schwarz and multiplying and dividing we get
    \begin{align*}
        \sqrt{\floor{n/2}}\sup_{\pi \in \Pi} &\biggl| \floor{n/2}^{-1} \sum_{i=1}^{\floor{n/2}}\varepsilon_i  \Gamma_{i,\floor{n/2} + i}^{ab}(\pi_{ab}(X_{i},X_{\floor{n/2} + i}) - b_{J}(\pi_{ab}(X_{i},X_{\floor{n/2} + i}))) \biggr| \\
        &\leq \sqrt{\floor{n/2}} \sup_{\pi \in \Pi}  \frac{\sqrt{\floor{n/2}^{-1} \sum_{i=1}^{\floor{n/2}} \biggl| \Gamma_{i,\floor{n/2} + i}^{ab}(\pi_{ab}(X_{i},X_{\floor{n/2} + i}) - b_{J}(\pi_{ab}(X_{i},X_{\floor{n/2} + i}))) \biggr|^2}}{\sqrt{\floor{n/2}^{-1} \sum_{i=1}^{\floor{n/2}} \Gamma_{i,\floor{n/2} + i}^{2 \, ab}}} \\
        &\times \sqrt{\floor{n/2}^{-1} \sum_{i=1}^{\floor{n/2}} \Gamma_{i,\floor{n/2} + i}^{2 \, ab}} \\
        &= \sqrt{\floor{n/2}}\sup_{\pi \in \Pi} D_n(\pi_{ab}, b_J(\pi_{ab})) \sqrt{\floor{n/2}^{-1} \sum_{i=1}^{\floor{n/2}} \Gamma_{i,\floor{n/2} + i}^{2 \, ab}} \\
        &\leq \sqrt{\floor{n/2}} 2^{-J} \sqrt{\floor{n/2}^{-1} \sum_{i=1}^{\floor{n/2}} \Gamma_{i,\floor{n/2} + i}^{2 \, ab}} = \frac{M}{\floor{n/2}^{1/2-\beta}} \to 0,
    \end{align*}
    last inequality uses Lemma 2 Online Appendix and the last equality uses $J = \ceil{\log_2(\floor{n/2})(1-\beta)}$ and the boundedness assumption. The second to last term of the Rademacher decomposition is also negligible. Notice that $\{\varepsilon_i  \Gamma_{i,\floor{n/2} + i}^{ab}(b_J(\pi_{ab}(X_{i},X_{\floor{n/2} + i})) - b_{\underline{J}}(\pi_{ab}(X_{i},X_{\floor{n/2} + i})))\}_{i=1}^{\floor{n/2}}$ are zero mean (conditional on $\{X_i,\Gamma_{i,\floor{n/2}+i}\}_{i=1}^{\floor{n/2}}$) i.i.d. random variables. They are also bounded below by $a_i = -|\Gamma_{i,\floor{n/2} + i}^{ab}(b_J(\pi_{ab}(X_{i},X_{\floor{n/2} + i})) - b_{\underline{J}}(\pi_{ab}(X_{i},X_{\floor{n/2} + i})))|$ and above by $b_i = -a_i$. Hence, by Hoeffding's inequality
    \begin{align*}
        \Pr_\varepsilon &\biggl(\biggl| \sum_{i=1}^{\floor{n/2}} \varepsilon_i  \Gamma_{i,\floor{n/2} + i}^{ab}(b_J(\pi_{ab}(X_{i},X_{\floor{n/2} + i})) - b_{\underline{J}}(\pi_{ab}(X_{i},X_{\floor{n/2} + i}))) \biggr| \geq t \biggr) \\
        &\leq 2\exp\biggl(-\frac{2t^2}{\sum_{i=1}^{\floor{n/2}}(b_i - a_i)^2}\biggr) = 2\exp\biggl(-\frac{t^2}{D_n^2(b_J(\pi_{ab}),b_{\underline{J}}(\pi_{ab}))\sum_{i=1}^{\floor{n/2}} \Gamma_{i,\floor{n/2} + i}^{2 \, ab}}\biggr). \\
    \end{align*}
    Hence, for any $a > 0$ we have that
    \begin{align*}
        \Pr_\varepsilon &\biggl(\biggl| \sqrt{\floor{n/2}} \floor{n/2}^{-1} \sum_{i=1}^{\floor{n/2}} \varepsilon_i  \Gamma_{i,\floor{n/2} + i}^{ab}(b_J(\pi_{ab}(X_{i},X_{\floor{n/2} + i})) - b_{\underline{J}}(\pi_{ab}(X_{i},X_{\floor{n/2} + i}))) \biggr| \\
        &\geq a 2^{2-\underline{J}} \sqrt{\frac{\sum_{i=1}^{\floor{n/2}} \Gamma_{i,\floor{n/2} + i}^{2 \, ab}}{\floor{n/2}}} \biggr) \\
        &\leq 2\exp\biggl(-\frac{a^2 4^{2-\underline{J}}}{D_n^2(b_J(\pi_{ab}),b_{\underline{J}}(\pi_{ab}))} \biggr) \\
        &\leq 2\exp\biggl(-\frac{a^2 4^{2-\underline{J}}}{\sum_{j = \underline{J}}^{J-1}D_n^2(b_j(\pi_{ab}),b_{j+1}(\pi_{ab}))} \biggr) \\
        &\leq 2\exp\biggl(-\frac{a^2 4^{2-\underline{J}}}{\biggl(\sum_{j = \underline{J}}^{J-1} 2^{-(j-1)}\biggr)^2} \biggr)\\
        &\leq 2 \exp(-a^2),
    \end{align*}
    where we have used triangle inequality in the second inequality and that $\sum_{j = \underline{J}}^{J-1} 2^{-(j-1)} = 2^{2-\underline{J}} - 2^{2-J} \leq 2^{2-\underline{J}}$ in the last inequality. This holds for any policy, hence
    \begin{align*}
        \Pr_\varepsilon &\biggl(\sup_{\pi \in \Pi}\biggl| \floor{n/2}^{-1/2} \sum_{i=1}^{\floor{n/2}} \varepsilon_i  \Gamma_{i,\floor{n/2} + i}^{ab}(b_J(\pi_{ab}(X_{i},X_{\floor{n/2} + i})) - b_{\underline{J}}(\pi_{ab}(X_{i},X_{\floor{n/2} + i}))) \biggr| \\
        &\geq a 2^{2-\underline{J}} \sqrt{\frac{\sum_{i=1}^{\floor{n/2}} \Gamma_{i,\floor{n/2} + i}^{2 \, ab}}{\floor{n/2}}} \biggr) \\
        &\leq 2 |\{b_J(\pi_{ab}),b_{\underline{J}}(\pi_{ab})\}| \exp(-a^2) \\
        &\leq 2 N_{D_n}(2^{-J},\T{\Pi}_{ab},\{X_i,\Gamma_{i,\floor{n/2}+i}\}_{i=1}^{\floor{n/2}}) \exp(-a^2) \\
        &\leq 2 N_H(2^{-2J},\T{\Pi}_{ab}) \exp(-a^2) \\
        &= 2 \exp(\log(N_H(2^{-2J},\T{\Pi}_{ab}))) \exp(-a^2) \\
        &\leq 2 \exp(5 VC(\T{\Pi}_{ab}) \log(2^{2J}) -a^2) \\
        &\leq 2\exp(5 VC(\T{\Pi}_{ab}) \log(2^{-2(1-\beta)\log_2(\floor{n/2})}) - a^2),
    \end{align*}
    where in the first inequality I use the union bound, in the second inequality I use properties of the approximations (see \cite{zhou2023offline}), in the third I use Lemma Lemma 1 in the Online Appendix and in the fourth inequality I bound the log of the Hamming covering number by the VC dimension using a result in \cite{haussler1995sphere}. Let now
    \[
        a = \frac{2^{\underline{J}}}{\sqrt{\log(\floor{n/2}) \floor{n/2}^{-1} \sum_{i=1}^{\floor{n/2}} \Gamma_{i,\floor{n/2} + i}^{2 \, ab}}} \text{ so that }  a 2^{2-\underline{J}} \sqrt{\frac{\sum_{i=1}^{\floor{n/2}} \Gamma_{i,\floor{n/2} + i}^{2 \, ab}}{\floor{n/2}}} = \frac{4}{\sqrt{\log(\floor{n/2})}}.
    \]
    Finally,
    \begin{align*}
        \Pr_\varepsilon &\biggl(\sup_{\pi \in \Pi}\biggl| \floor{n/2}^{-1/2} \sum_{i=1}^{\floor{n/2}} \varepsilon_i  \Gamma_{i,\floor{n/2} + i}^{ab}(b_J(\pi_{ab}(X_{i},X_{\floor{n/2} + i})) - b_{\underline{J}}(\pi_{ab}(X_{i},X_{\floor{n/2} + i}))) \biggr| \\
        &\geq\frac{4}{\sqrt{\log(\floor{n/2})}} \biggr) \leq 2\exp\biggl(5 VC(\T{\Pi}_{ab}) \log(\floor{n/2}^{-2(1-\beta)}) - \frac{\floor{n/2}^{-\beta}}{\log(\floor{n/2}) \sum_{i=1}^{\floor{n/2}} \Gamma_{i,\floor{n/2} + i}^{2 \, ab}}\biggr) \\
        &\leq 2\exp\biggl\{-5 \floor{n/2}^\beta \log\biggl(\floor{n/2}^{2(1-\beta)}\biggr) - \frac{1}{\floor{n/2}^\beta \log(\floor{n/2})M^2} \biggr\} \to 0,
    \end{align*}
    where I have used Assumption \ref{ass_VC} and the boundedness assumption. 
    \begin{align*}
        \Ex\biggl(\sup_{\pi \in \Pi}\biggl| \floor{n/2}^{-1/2} \sum_{i=1}^{\floor{n/2}} \varepsilon_i  \Gamma_{i,\floor{n/2} + i}^{ab}(b_J(\pi_{ab}(X_{i},X_{\floor{n/2} + i})) - b_{\underline{J}}(\pi_{ab}(X_{i},X_{\floor{n/2} + i}))) \biggr| \biggr) \to 0,
    \end{align*}
    since for any sequence of random variables $X_n$ and sequence of real numbers $a_n$ if $\lim_{n \to \infty} \Pr(X_n \leq a_n) = 1$ and $\lim_{n \to \infty}a_n = 0$, then $\lim_{n \to \infty} \Ex(X_n) = 0$ (proof of this fact uses $\Ex(X_n) = \int_0^\infty \Pr(X_n > u) \, du$). Hence, we have proven that
    \begin{align*}
        \Ex[\mathcal{R}_n(\Pi)] &= \Ex \biggl\{\sup_{\pi \in \Pi} \biggl| \floor{n/2}^{-1} \sum_{i=1}^{\floor{n/2}}\varepsilon_i  \Gamma_{i,\floor{n/2} + i}^{ab}\biggl[\sum_{j=1}^{\underline{J}} \biggl(b_j(\pi_{ab}(X_{i},X_{\floor{n/2} + i})) - b_{j-1}(\pi_{ab}(X_{i},X_{\floor{n/2} + i}))\biggr) \biggr] \biggr| \biggr\} \\
        &+ o\biggl(\frac{1}{\sqrt{n}}\biggr).
    \end{align*} 
    Hence I have left what \cite{zhou2023offline} call the effective regime. Let $j \in \{1,...,\underline{J}\}$ and $a_j$ be some constant depending on $j$. As before, conditional on $\{X_i,\Gamma_{i,\floor{n/2}+i}\}_{i=1}^{\floor{n/2}}$ we can apply Hoeffding inequality and then use the definition of $D_n$ to get
    \begin{align*}
        \Pr_\varepsilon &\biggl(\biggl| \floor{n/2}^{-1/2} \sum_{i=1}^{\floor{n/2}} \varepsilon_i \Gamma_{i,\floor{n/2} + i}^{ab}(b_j(\pi_{ab}(X_{i},X_{\floor{n/2} + i})) - b_{j-1}(\pi_{ab}(X_{i},X_{\floor{n/2} + i}))) \biggr| \\
        & \geq a_j 2^{2-j} \sqrt{\frac{\sum_{i=1}^{\floor{n/2}} \Gamma_{i,\floor{n/2} + i}^{2 \, ab}}{\floor{n/2}}} \biggr) \\
        &\leq 2\exp\biggl(-\frac{a_j^2 4^{2-j}}{D_n^2(b_j(\pi_{ab}),b_{j-1}(\pi_{ab}))} \biggr) \\
        &\leq 2 \exp\biggl(\frac{-a_j^2 4^{2-j}}{4^{-(j-1)}}\biggr) \\
        &= 2 \exp\biggl(-4a_j^2 \biggr),
    \end{align*}
    since $D_n(b_j(\pi_{ab}),b_{j-1}(\pi_{ab})) \leq 2^{-(j-1)}$ by Lemma 2 Online Appendix. Now we let
    \[
        a_j^2(k) = 2 \log \biggl(\frac{2j^2}{\delta_k} N_H(4^{-j},\T{\Pi}_{ab})\biggr),  
    \]
    where $\delta_k$ is some sequence of real numbers indexed by $k \in \mathbb{N}$. Define
    \[
        R_j = \sup_{\pi \in \Pi} \biggl| \floor{n/2}^{-1/2} \sum_{i=1}^{\floor{n/2}} \varepsilon_i \Gamma_{i,\floor{n/2} + i}^{ab}(b_j(\pi_{ab}(X_{i},X_{\floor{n/2} + i})) - b_{j-1}(\pi_{ab}(X_{i},X_{\floor{n/2} + i}))) \biggr|.
    \]  
    Then we have that
    \begin{align*}
        \Pr\biggl(R_j \geq a_j(k) 2^{-j} \sqrt{\floor{n/2}^{-1} \sum_{i=1}^{\floor{n/2}} \Gamma_{i,\floor{n/2} + i}^{2 \, ab}}\biggr) &\leq 2 |\{b_j(\pi_{ab}),b_{j-1}(\pi_{ab})\}| \exp(-a_j^2(k)/2) \\
        &\leq 2 N_{D_n}(2^{-j},\T{\Pi}_{ab},\{X_i,\Gamma_{i,\floor{n/2}+i}\}_{i=1}^{\floor{n/2}}) \exp(-a_j^2(k)/2) \\
        &\leq 2 N_H(2^{-2j},\T{\Pi}_{ab}) \exp(-a_j^2(k)/2) \\
        &= 2 N_H(4^{-j},\T{\Pi}_{ab}) \exp(-\log(N_H(4^{-j},\T{\Pi}_{ab}) 2j^2/\delta_k)) \\
        &= \frac{\delta_k}{j^2}.
    \end{align*}
    Sum across $j$ and apply this bound with $\delta_k = 1/2^k$ to note that 
    \begin{align*}
        \sum_{j=1}^{\underline{J}} \Pr\biggl(R_j \geq a_j(k) 2^{-j} \sqrt{\floor{n/2}^{-1} \sum_{i=1}^{\floor{n/2}} \Gamma_{i,\floor{n/2} + i}^{2 \, ab}}\biggr) &\leq \sum_{j=1}^{\underline{J}} \frac{\delta_k}{j^2} \\
        &\leq \sum_{j=1}^{\infty} \frac{\delta_k}{j^2} \\
        &\leq \frac{1.7}{2^k}.
    \end{align*}
    Let $F_{R_j}$ be the cumulative distribution function of $R_j$ (conditional on $\{X_i,\Gamma_{i,\floor{n/2}+i}\}_{i=1}^{\floor{n/2}}$). We can bound the following object of interest in the following way
    \begin{align*}
        \Ex_\varepsilon&\biggl[\sup_{\pi \in \Pi} \biggl| \floor{n/2}^{-1/2} \sum_{i=1}^{\floor{n/2}} \varepsilon_i \Gamma_{i,\floor{n/2} + i}^{ab}\sum_{j=1}^{\underline{J}}(b_j(\pi_{ab}(X_{i},X_{\floor{n/2} + i})) - b_{j-1}(\pi_{ab}(X_{i},X_{\floor{n/2} + i}))) \biggr|\biggr] \\
        &\leq \sum_{j=1}^{\underline{J}} \Ex_\varepsilon[R_j] =  \int_0^\infty \sum_{j=1}^{\underline{J}} (1-F_{R_j}(r)) \, dr \leq \int_0^\infty \sum_{j=1}^{\underline{J}} \Pr(R_j \geq r) \, dr \\
        &\leq \sum_{k=0}^\infty \sum_{j=1}^{\underline{J}} \frac{1.7}{2^k} a_j(k) 2^{-j}\sqrt{\floor{n/2}^{-1} \sum_{i=1}^{\floor{n/2}} \Gamma_{i,\floor{n/2} + i}^{2 \, ab}} \\
        &\leq \sum_{k=0}^\infty \sum_{j=1}^{\underline{J}} \frac{1.7}{2^k} \sqrt{2} \sqrt{\log(2^{k+1} j^2 N_H(4^{-j},\T{\Pi}_{ab}))}2^{-j} \sqrt{\floor{n/2}^{-1} \sum_{i=1}^{\floor{n/2}} \Gamma_{i,\floor{n/2} + i}^{2 \, ab}} \\
        &\leq 1.7 \sqrt{2} \sqrt{\floor{n/2}^{-1} \sum_{i=1}^{\floor{n/2}} \Gamma_{i,\floor{n/2} + i}^{2 \, ab}} \sum_{k=0}^\infty 2^{-k} \sum_{j=1}^{\underline{J}} 2^{-j} \sqrt{(k+1)\log 2 + 2 \log j + \log N_H(4^{-j},\T{\Pi}_{ab})} \\
        &\leq 1.7 \sqrt{2} \sqrt{\floor{n/2}^{-1} \sum_{i=1}^{\floor{n/2}} \Gamma_{i,\floor{n/2} + i}^{2 \, ab}} \sum_{k=0}^\infty 2^{-k} \sum_{j=1}^{\underline{J}} 2^{-j} \biggl(\sqrt{k+1} + \sqrt{2 \log j} + \sqrt{5 VC(\T{\Pi}_{ab}) \log(4^j)}\biggr) \\
        &\leq 1.7 \sqrt{2} \sqrt{\floor{n/2}^{-1} \sum_{i=1}^{\floor{n/2}} \Gamma_{i,\floor{n/2} + i}^{2 \, ab}} \sum_{k=0}^\infty 2^{-k} \biggl( \sqrt{k+1} \sum_{j=1}^\infty 2^{-j} + \sqrt{2} \sum_{j=1}^{\infty} 2^{-j} \sqrt{\log j} \\
        &+ \sqrt{5 VC(\T{\Pi}_{ab})} \sum_{j=1}^\infty 2^{-j} \sqrt{\log 4^{j}}\biggr) \\
        &\leq 1.7 \sqrt{2} \sqrt{\floor{n/2}^{-1} \sum_{i=1}^{\floor{n/2}} \Gamma_{i,\floor{n/2} + i}^{2 \, ab}} \biggl( \sum_{k=0}^\infty 2^{-k} \sqrt{k+1} + \frac{\sqrt{2}}{2} \sum_{k=0}^\infty 2^{-k} + \sqrt{5 VC(\T{\Pi}_{ab})} 1.6 \sum_{k= 0}^\infty 2^{-k} \biggr) \\
        &\leq 1.7 \sqrt{2} \sqrt{\floor{n/2}^{-1} \sum_{i=1}^{\floor{n/2}} \Gamma_{i,\floor{n/2} + i}^{2 \, ab}} \biggl(5 + 3.2 \sqrt{5 VC(\T{\Pi}_{ab})}\biggr).
    \end{align*}
    So taking expectations over $\{X_i,\Gamma_{i,\floor{n/2}+i}\}_{i=1}^{\floor{n/2}}$, using this bound and the Jensen's inequality we get
    \begin{align*}
        \Ex&\biggl[\sup_{\pi \in \Pi} \biggl| \floor{n/2}^{-1/2} \sum_{i=1}^{\floor{n/2}} \varepsilon_i \Gamma_{i,\floor{n/2} + i}^{ab}\sum_{j=1}^{\underline{J}}(b_j(\pi_{ab}(X_{i},X_{\floor{n/2} + i})) - b_{j-1}(\pi_{ab}(X_{i},X_{\floor{n/2} + i}))) \biggr|\biggr] \\
        &\leq 1.7 \sqrt{2} \biggl(5 + 8 \sqrt{5 VC(\T{\Pi}_{ab})}\biggr) \Ex \biggl[ \sqrt{\floor{n/2}^{-1} \sum_{i=1}^{\floor{n/2}} \Gamma_{i,\floor{n/2} + i}^{2 \, ab}} \biggr] \\
        &\leq 1.7 \sqrt{2} \biggl(5 + 8 \sqrt{5 VC(\T{\Pi}_{ab})}\biggr) \sqrt{\Ex \biggl[ \Gamma_{i,\floor{n/2} + i}^{2 \, ab} \biggr]} \\
        &=  1.7 \sqrt{2} \biggl(5 + 8 \sqrt{5 VC(\T{\Pi}_{ab})}\biggr) \sqrt{S_{ab}} \leq C\sqrt{VC(\T{\Pi}_{ab}) S_{ab}},
    \end{align*}
    for some constant $C > 0$. Dividing both sides by $\sqrt{\floor{n/2}}$ we get 
    \begin{align*}
        \Ex&\biggl[\sup_{\pi \in \Pi} \biggl| \floor{n/2}^{-1} \sum_{i=1}^{\floor{n/2}} \varepsilon_i \Gamma_{i,\floor{n/2} + i}^{ab}\sum_{j=1}^{\underline{J}}(b_j(\pi_{ab}(X_{i},X_{\floor{n/2} + i})) - b_{j-1}(\pi_{ab}(X_{i},X_{\floor{n/2} + i}))) \biggr|\biggr] \\
        &\leq C\sqrt{\frac{VC(\T{\Pi}_{ab}) S_{ab}}{\floor{n/2}}},
    \end{align*}
    and hence $\Ex[\mathcal{R}_n(\Pi)] \leq C\sqrt{\frac{VC(\T{\Pi}_{ab}) S_{ab}}{\floor{n/2}}} + o\biggl(\frac{1}{\sqrt{n}}\biggr) = \mathcal{O}\biggl(\sqrt{\frac{VC(\T{\Pi}_{ab}) S_{ab}}{\floor{n/2}}}\biggr)$.
\end{proof}

\begin{proof}[Proof of Lemma \ref{lemma_uniform_coupling}]
    Define the following random variables
    \begin{align*}
        \hat{R}_{ij,ab,l}^{(1)} &= m_{ab}(Z_i,Z_j,\hat{\gamma}_l,\varphi) - m_{ab}(Z_i,Z_j,\gamma,\varphi), \quad  \hat{R}_{ij,ab,l}^{(2)} &= m_{ab}(Z_i,Z_j,\gamma,\hat{\varphi}_l) - m_{ab}(Z_i,Z_j,\gamma,\varphi) \\
        \hat{R}_{ij,ab,l}^{(3)} &= \phi_{ab}^{\gamma}(Z_i,Z_j,\hat{\gamma}_l,\alpha^{\gamma}) - \phi_{ab}^{\gamma}(Z_i,Z_j,\gamma,\alpha^{\gamma}), \quad \hat{R}_{ij,ab,l}^{(4)} &= \phi_{ab}^{\gamma}(Z_i,Z_j,\gamma,\hat{\alpha}_l^{\gamma}) - \phi_{ab}^{\gamma}(Z_i,Z_j,\gamma,\alpha^{\gamma}) \\
        \hat{R}_{ij,ab,l}^{(5)} &= \phi_{ab}^{\varphi}(Z_i,Z_j,\hat{\varphi}_l,\alpha^{\varphi}) - \phi_{ab}^{\varphi}(Z_i,Z_j,\varphi,\alpha^{\varphi}), \quad \hat{R}_{ij,ab,l}^{(6)} &= \phi_{ab}^{\varphi}(Z_i,Z_j,\varphi,\hat{\alpha}_l^{\varphi}) - \phi_{ab}^{\varphi}(Z_i,Z_j,\varphi,\alpha^{\varphi}).
    \end{align*}
Then,
\footnotesize
\begin{align*}
    \Ex\biggl[\sup_{\pi \in \Pi_n}|\hat{W}_n(\pi) - \T{W}_n(\pi)|\biggr] &= 
    \Ex\biggl(\sup_{\pi \in \Pi_n}\biggl| \binom{n}{2}^{-1} \sum_{l=1}^L \sum_{(i,j)\in I_l} \sum_{(a,b)\in \{0,1\}^2} \sum_{k=1}^6 (\hat{R}_{ij,ab,l}^{(k)} + \hat{\xi}_{ij,ab,l} + \hat{\xi}_{ij,ab,l}^{\gamma} + \hat{\xi}_{ij,ab,l}^{\varphi})\pi_{ab}(X_i,X_j)\biggr|\biggr).
\end{align*}
\normalsize
By repeated use of the triangle inequality
\begin{align*}
    (\dagger) \quad \Ex\biggl[\sup_{\pi \in \Pi_n}|\hat{W}_n(\pi) - \T{W}_n(\pi)|\biggr] 
    &\leq \sum_{l=1}^L \sum_{(a,b)\in \{0,1\}^2}\Ex\biggl(\sup_{\pi \in \Pi_n}\biggl| \binom{n}{2}^{-1} \sum_{(i,j)\in I_l}  (\hat{R}_{ij,l}^{(1)} + \hat{R}_{ij,l}^{(3)})\pi_{ab}(X_i,X_j) \biggr|\biggr) \\
    &\quad + \sum_{l=1}^L \sum_{(a,b)\in \{0,1\}^2} \Ex\biggl(\sup_{\pi \in \Pi_n}\biggl|\binom{n}{2}^{-1}  \sum_{(i,j)\in I_l} (\hat{R}_{ij,l}^{(2)} + \hat{R}_{ij,l}^{(5)})\pi_{ab}(X_i,X_j) \biggr|\biggr) \\
    &\quad + \sum_{l=1}^L \sum_{(a,b)\in \{0,1\}^2} \Ex\biggl(\sup_{\pi \in \Pi_n}\biggl| \binom{n}{2}^{-1} \sum_{(i,j)\in I_l} (\hat{R}_{ij,l}^{(4)} + \hat{R}_{ij,l}^{(6)})\pi_{ab}(X_i,X_j) \biggr|\biggr) \\
    &\quad + \sum_{l=1}^L \sum_{(a,b)\in \{0,1\}^2} \Ex\biggl(\sup_{\pi \in \Pi_n}\biggl| \binom{n}{2}^{-1} \sum_{(i,j)\in I_l} (\hat{\xi}_{ij,l} + \hat{\xi}_{ij,l}^{\gamma} + \hat{\xi}_{ij,l}^{\varphi})\pi_{ab}(X_i,X_j) \biggr|\biggr).
\end{align*}
Without loss of generality I focus on some fixed $(a,b)$ and $l$. Let $N_l^c$ be the observations not in $I_l$. By adding and subtracting $\Ex[(\hat{R}_{ij,ab,l}^{(1)} + \hat{R}_{ij,ab,l}^{(3)})\pi_{ab}(X_i, X_j)|N_l^c]$ and triangle inequality, the summands of the first term are bounded by

\begin{align*}
    \Ex\biggl(&\sup_{\pi \in \Pi_n}\biggl| \binom{n}{2}^{-1} \sum_{(i,j)\in I_l}  (\hat{R}_{ij,l}^{(1)} + \hat{R}_{ij,l}^{(3)})\pi_{ab}(X_i,X_j) - \Ex[(\hat{R}_{ij,ab,l}^{(1)} + \hat{R}_{ij,ab,l}^{(3)})\pi_{ab}(X_i,X_j)|N_l^c]) \biggr|\biggr) \quad (\star) \\
    &+ \Ex\biggl(\sup_{\pi \in \Pi_n} \binom{n}{2}^{-1} \sum_{(i,j)\in I_l}  |\Ex[(\hat{R}_{ij,ab,l}^{(1)} + \hat{R}_{ij,ab,l}^{(3)})\pi_{ab}(X_i,X_j)|\hat{\gamma}_l]| \biggr). \quad (\star \star)
\end{align*}
\normalsize
By Assumption \ref{ass_glob_rob_quadr} we know that 
\begin{align*}
    |\Ex[\hat{R}_{ij,ab,l}^{(1)} + \hat{R}_{ij,ab,l}^{(3)}|N_l^c]| &= |\Ex[m_{ab}(Z_i,Z_j,\hat{\gamma}_l,\varphi) + \phi_{ab}^{\gamma}(Z_i,Z_j,\hat{\gamma}_l,\alpha^{\gamma})|\hat{\gamma}_l]| \leq C||\hat{\gamma}_l - \gamma||^2.
\end{align*}
Applying the conditional Jensen's inequality (on the absolute value) in $(\star \star)$ and noting that the resulting expression is maximized by treating everybody we get that 
\begin{align*}
    (\star \star) &\leq C \Ex[||\hat{\gamma}_l - \gamma||^2] \underbrace{\binom{n}{2}^{-1}|I_l|}_{\leq 1} = o(n^{-2\lambda_\gamma}) = o(1/\sqrt{n}),
\end{align*}
where the last equality follows since $2\lambda_\gamma \geq 1/2$. For $(\star)$, note that
\begin{align*}
    (\star) &\leq \binom{n}{2}^{-1}|I_l|\Ex\biggl(\sup_{\pi \in \Pi_n}\biggl| |I_l|^{-1}  \sum_{(i,j)\in I_l}  \hat{R}_{ij,l}^{(1)}\pi_{ab}(X_i,X_j) - \Ex[\hat{R}_{ij,ab,l}^{(1)}\pi_{ab}(X_i,X_j)|N_l^c] \biggr|\biggr) \\
    &+ \binom{n}{2}^{-1}|I_l|\Ex\biggl(\sup_{\pi \in \Pi_n}\biggl| |I_l|^{-1} \sum_{(i,j)\in I_l}  \hat{R}_{ij,l}^{(3)}\pi_{ab}(X_i,X_j) - \Ex[\hat{R}_{ij,ab,l}^{(3)}\pi_{ab}(X_i,X_j)|N_l^c]) \biggr|\biggr) \\
    &= \binom{n}{2}^{-1}|I_l|\Ex\biggl[\Ex\biggl(\sup_{\pi \in \Pi_n}\biggl| |I_l|^{-1} \sum_{(i,j)\in I_l}  \hat{R}_{ij,l}^{(1)}\pi_{ab}(X_i,X_j) - \Ex[\hat{R}_{ij,ab,l}^{(1)}\pi_{ab}(X_i,X_j)|N_l^c] \biggr| \, \biggr| N_l^c\biggr)\biggr] \\
    &+ \binom{n}{2}^{-1}|I_l|\Ex\biggl[\Ex\biggl(\sup_{\pi \in \Pi_n}\biggl| |I_l|^{-1} \sum_{(i,j)\in I_l}  \hat{R}_{ij,l}^{(3)}\pi_{ab}(X_i,X_j) - \Ex[\hat{R}_{ij,ab,l}^{(3)}\pi_{ab}(X_i,X_j)|N_l^c] \biggr| \, \biggr| N_l^c \biggr)\biggr].
\end{align*}
The inner expectations are the expected supremum of centered U-processes. Using Lemma \ref{lemma_rademacher} I can bound them with Rademacher complexities. However, in the same way we used the construction in Section A of the Online Appendix in Lemma \ref{lemma_rademacher} to be able to bound the U-process with a Rademacher complexity which involves a sum of independent terms, we need to use such a construction for each fold $I_l$. Take the cross-fitting technique in \cite{escanciano2023machine} where we split $\{1,...,n\}$ into sets $\mathcal{C} = \{C_1,...,C_K\}$ and take the intersection between $\mathcal{C}^2$ and the set $\{(i,j) \in \{1,...,n\}^2: i<j\}$. $I_l$ can be either a triangle ($I_l \in T$, where $T = \{I_l: i \in C_f, j \in C_g, f < g, (i,j) \in I_l\}$) or a square ($I_l \in S$, where $S = \{I_l: i \in C_f, j \in C_g, f = g, (i,j) \in I_l\}$) and that in each case we can bound the U-process with the following Rademacher complexities
\begin{align*}
    \mathcal{R}_{n,l}(\Pi_{ab}) &= \begin{cases}
        \Ex_\varepsilon \biggl(\sup_{\pi \in \Pi} \biggl| |C_k|^{-1} \sum_{i=1}^{|C_k|} \varepsilon_i \hat{R}_{\rho(i,k),|C_k|+i}^{(q)}\pi_{ab}(X_{\rho(i,k)},X_{|C_k| + i}) \biggr| \biggr) &\text{if } I_l \in S \\
        \Ex_\varepsilon \biggl(\sup_{\pi \in \Pi} \biggl| \floor{|C_k|/2}^{-1} \sum_{i=1}^{\floor{|C_k|/2}}\varepsilon_i  \hat{R}_{i,\floor{|C_k|/2} + i,l}^{(q)} \pi_{ab}(X_{i},X_{\floor{|C_k|/2} + i})\biggr| \biggr) &\text{if } I_l \in T,
        \end{cases} 
\end{align*}
for $q = 1,3$. Hence, by Lemmas \ref{lemma_rademacher} and \ref{lemma_boundforrademacher} we have that
\begin{align*}
    \binom{n}{2}^{-1}|I_l|\Ex&\biggl(\sup_{\pi \in \Pi_n}\biggl| |I_l|^{-1} \sum_{(i,j)\in I_l}  \hat{R}_{ij,l}^{(1)}\pi_{ab}(X_i,X_j) - \Ex[\hat{R}_{ij,ab,l}^{(1)}\pi_{ab}(X_i,X_j)|N_l^c] \biggr| \, \biggr| N_l^c\biggr) \\
    &= \mathcal{O}\biggl(\sqrt{\frac{S_{ab,l}^{(1)}VC(\Pi_{ab,n})}{\floor{C_k/2}}}\biggr),
\end{align*}
where $S_{ab,l}^{(1)} = \Ex[\hat{R}_{ij,l}^{(1)^2}|N_l^c]$. Noting that $\Ex[S_{ab,l}^{(1)}] = \Ex[(m_{ab}(Z_i,Z_j,\hat{\gamma}_l,\varphi) - m_{ab}(Z_i,Z_j,\gamma,\varphi))^2]$ and using Assumption \ref{ass_nuisance_rates}, Jensen's inequality, the fact that $|I_l| = |C_k|\times |C_m|$ if $I_l = I(C_k,C_m)$ and $|I_l| = |C_k|\times |C_k - 1|/2$ if $I_l = I(C_k,C_k)$ and that for evenly sized folds $|C_k|/(n-1) \leq 1$ for all $k = 1,...,K$ we have that
\begin{align*}
    \Ex\biggl(&\sup_{\pi \in \Pi_n}\biggl| \binom{n}{2}^{-1} \sum_{(i,j)\in I_l}  \hat{R}_{ij,l}^{(1)}\pi_{ab}(X_i,X_j) - \Ex[\hat{R}_{ij,ab,l}^{(1)}\pi_{ab}(X_i,X_j)|N_l^c]) \biggr|\biggr) \\
    &= \mathcal{O}\biggl(\sqrt{VC(\Pi_{ab,n})\frac{a((1-K^{-1})n)^2}{n^{1 + 2\lambda_\gamma}}}\biggr).
\end{align*}
The same bound applies by using the same arguments when we replace $\hat{R}_{ij,l}^{(1)}$ by $\hat{R}_{ij,l}^{(3)}$. This bound applies to all folds $I_l$, so, summing across all folds gives us the same asymptotic bound. Hence, we have bounded the first term on the right-hand side in $(\dagger)$. For the second term, we follow the same steps as with the first term to get the same bounds with $\lambda_\gamma$ replaced by $\lambda_\varphi$. For the third term in $(\dagger)$, by Assumption \ref{ass_glob_rob_quadr} (i) (global robustness of $\alpha$), $\Ex[\hat{R}_{ij,ab,l}^{(4)}|N_l^c] = \Ex[\hat{R}_{ij,ab,l}^{(6)}|N_l^c] = 0$. Thus, we can apply Lemmas \ref{lemma_rademacher} and \ref{lemma_boundforrademacher} directly to get that for $q = 4, 6$
\[
    \Ex\biggl(\sup_{\pi \in \Pi_n}\biggl| \binom{n}{2}^{-1} \sum_{(i,j)\in I_l}  \hat{R}_{ij,l}^{(q)}\pi_{ab}(X_i,X_j) \biggr|\biggr) = \mathcal{O}\biggl(\sqrt{VC(\Pi_{ab,n})\frac{a((1-K^{-1})n)^2}{n^{1 + 2\lambda_\alpha}}}\biggr).    
\]
Finally, the bound for the last term in ($\dagger$) follows directly from Assumption \ref{ass_interaction_term}
\[
    \Ex\biggl(\sup_{\pi \in \Pi_n}\biggl| \binom{n}{2}^{-1} \sum_{(i,j)\in I_l} (\hat{\xi}_{ij,l} + \hat{\xi}_{ij,l}^{\gamma} + \hat{\xi}_{ij,l}^{\varphi})\pi_{ab}(X_i,X_j) \biggr|\biggr) = \mathcal{O}\biggl(\frac{a(1-K^{-1})}{\sqrt{n}}\biggr).
\]
Putting everything together we know that 
\begin{align*}
    \sqrt{n}\Ex\biggl[\sup_{\pi \in \Pi_n}|\hat{W}_n(\pi) - \T{W}_n(\pi)|\biggr] &= \mathcal{O}\biggl(\sqrt{VC(\Pi_{ab,n})\frac{a((1-K^{-1})n)^2}{n^{2\lambda_\gamma}}}\biggr) \\
    &+ \mathcal{O}\biggl(\sqrt{VC(\Pi_{ab,n})\frac{a((1-K^{-1})n)^2}{n^{2\lambda_\varphi}}}\biggr) \\
    &+ \mathcal{O}\biggl(\sqrt{VC(\Pi_{ab,n})\frac{a((1-K^{-1})n)^2}{n^{2\lambda_\alpha}}}\biggr) \\
    &+ \mathcal{O}\biggl(a(1-K^{-1})\biggr) + o(1) \\
    &= \mathcal{O}\biggl(a((1-K^{-1})n)\biggl(1 + \sqrt{\frac{VC(\Pi_{ab,n})}{n^{2\min(\lambda_\gamma,\lambda_\varphi,\lambda_\alpha)}}} \biggr) \biggr).
\end{align*}
\end{proof}

\begin{proof}[Proof of Theorem \ref{thm_regret_upper_bound}]
    Follows from Lemmas \ref{lemma_boundforrademacher}, \ref{lemma_uniform_coupling} and Lemma 3 in the Online Appendix.
\end{proof}

\section{Online Appendix}
\setcounter{lemma}{0}
\subsection{Useful U-statistics representation}
I introduce a representation of U-statistics which will be very useful for the coming proofs. For any function $f: \mathcal{Z}^2 \to \mathbb{R}$ let $\mathbb{U}_n f(X_i,X_j) = \binom{n}{2}^{-1}\sum_{i<j}f(X_i,X_j)$. Let $\kappa$ be the permutations of $\{1,...,n\}$, then, as in \cite{hoeffding1963}, we can rewrite 
\begin{equation}
\label{eq_sum_iid_repr}
\mathbb{U}_n f(Z_i,Z_j) = \frac{1}{n!}\sum_\kappa \floor{n/2}^{-1} \sum_{i=1}^{\floor{n/2}}f(Z_{\kappa(i)},Z_{\kappa(\floor{n/2} + i)}).
\end{equation}
This expresses $\mathbb{U}_n f(Z_i,Z_j)$ as a (dependent) sum of averages of i.i.d. random variables (i.e. $f(Z_{\kappa(i)},Z_{\kappa(\floor{n/2} + i)})$ are i.i.d. for $i = 1,..., \floor{n/2}$). 
\subsection{Auxiliary lemmas}
\label{app_aux_lemmas}
Now some lemmas needed for the main results. For a fixed $\{X_i,\Gamma_{i,\floor{n/2}+i}\}_{i=1}^{\floor{n/2}}$ define
\begin{align*}
    \T{\Pi}_{ab} = \{\pi_{ab}(X_1,X_{\floor{n/2} + 1}),...,\pi_{ab}(X_{\floor{n/2}},X_{n}): \pi \in \Pi\}.
\end{align*}
For $\pi,\pi' \in \T{\Pi}_{ab}$ define the following distances
\begin{align*}
  D_n^2(\pi,\pi') &= \frac{\sum_{i=1}^{\floor{n/2} }\Gamma_{i,\floor{n/2}+i}^{2 \, ab}(\pi_{ab}(X_i,X_{\floor{n/2}+i}) - \pi_{ab}'(X_i,X_{\floor{n/2}+i}))^2}{\sum_{i=1}^{\floor{n/2} }\Gamma_{i,\floor{n/2}+i}^{2 \, ab}},  \\
  H(\pi,\pi') &= \frac{1}{n} \sum_{i = 1}^{n} \ind(\pi_{ab}(X_i,X_{\floor{n/2}+i}) \neq \pi_{ab}'(X_i,X_{\floor{n/2}+i})).
\end{align*}
Let $N_{D_n}(\varepsilon,\T{\Pi}_{ab},\{X_i,\Gamma_{i,\floor{n/2}+i}\}_{i=1}^{\floor{n/2}})$ be the number of balls of radius $\varepsilon$ needed to cover $\T{\Pi}_{ab}$ under distance $D_n$. Define the same object for the Hamming distance $H$ and let 
\[
N_H(\varepsilon, \T{\Pi}_{ab}) = \sup\{N_H(\varepsilon,\T{\Pi}_{ab},\{X_i\}_{i=1}^m): X_1,...,X_m \in \mathcal{X}, m \geq 1\}.
\]
Note $N_H(\varepsilon,\T{\Pi}_{ab})$ does not depend on $m$. It will be useful to bound $N_{D_n}$ with $N_H$.

\begin{lemma}
    \label{app_auxlemma_NDNH}
    For fixed $\{X_i,\Gamma_{i,\floor{n/2}+i}\}_{i=1}^{\floor{n/2}}$ we have that $N_{D_n}(\varepsilon,\T{\Pi}_{ab},\{X_i,\Gamma_{i,\floor{n/2}+i}\}_{i=1}^{\floor{n/2}}) \leq  N_H(\varepsilon^2, \T{\Pi}_{ab})$.
\end{lemma}

\begin{proof}
    Take an auxiliary sample $\{X_j'\}_{j=1}^m$ contained in $\{X_i\}_{i=1}^n$ such that
    \[
      \biggl| |B_i| - \frac{m \Gamma_{i,\floor
      {n/2}+i}^{2 \, ab}}{\sum_{i=1}^{\floor
      n/2}\Gamma_{i,\floor
      {n/2}+i}^{2 \, ab}} \biggr| \leq 1,
    \]
    where $B_i = \{j\in \{1,...,m\}: X_j' = X_i\}$. Then, for $\pi,\pi' \in \T{\Pi}_{ab}$
    \begin{align*}
        D_n^2(\pi,\pi') = \frac{1}{m}\sum_{i=1}^{\floor{n/2}}\underbrace{\frac{m\Gamma_{i,\floor
        {n/2}+i}^{2 \, ab}}{\sum_{k=1}^{\floor{n/2}}\Gamma_{k,\floor
        {n/2}+k}^{2 \, ab}}}_{\geq |B_i| - 1}\ind(\pi_{ab}(X_i,X_{\floor
        {n/2}+i}) \neq \pi_{ab}'(X_i,X_{\floor
        {n/2}+i}) ).
    \end{align*}
    So 
    \begin{align*}
        D_n^2(\pi,\pi') &\geq \sum_{i=1}^{\floor{n/2}} \frac{|B_i|}{m}\ind(\pi_{ab}(X_i,X_{\floor
        {n/2}+i}) \neq \pi_{ab}'(X_i,X_{\floor
        {n/2}+i})) - O(1/m) \\
        &= \sum_{i=1}^{\floor{n/2}} \frac{|B_i|}{m} \frac{1}{|B_i|}\sum_{j\in B_i}\ind(\pi_{ab}(X_j',X_{\floor
        {n/2}+j}') \neq \pi_{ab}'(X_j',X_{\floor
        {n/2}+j}')) - O(1/m) \\
        &= \frac{1}{m}\sum_{i=1}^{\floor{n/2}}\sum_{j\in B_i}\ind(\pi_{ab}(X_j',X_{\floor
        {n/2}+j}') \neq \pi_{ab}'(X_j',X_{\floor
        {n/2}+j}')) - O(1/m).
    \end{align*}
    The first equality uses that all summands in the inner sum are the same since for all $j \in B_i$ we know that $(X_i, X_{\floor
    {n/2}+i}) = (X_j',X_{\floor
    {n/2}+j}')$. The sum $\sum_{i=1}^{\floor{n/2}}\sum_{j\in B_i}$ can sum pairs more than once (e.g. if $(X_1,X_{\floor
    {n/2}+1}) = (X_2,X_{\floor
    {n/2}+2})$ then $B_1 = B_2$). Since $\{X'_j\}_{j=1}^m$ is contained in $\{X'_i\}_{i=1}^m$ 
    \begin{align*}
        D_n^2(\pi,\pi') &\geq \frac{1}{m}\sum_{j=1}^m \ind(\pi_{ab}(X_j',X_{\floor
        {n/2}+j}') \neq \pi_{ab}'(X_j',X_{\floor
        {n/2}+j}')) - O(1/m) \\
        &= H(\pi,\pi') - O(1/m).
    \end{align*} 
    So, $H(\pi,\pi') \leq D_n^2(\pi,\pi') + O(1/m)$. $N_H$ does not depend on $m$, so making $m$ large we conclude
    \[
        N_{D_n}(\varepsilon,\T{\Pi}_{ab},\{X_i,\Gamma_{i,\floor{n/2}+i}\}_{i=1}^{\floor{n/2}}) \leq  N_H(\varepsilon^2, \T{\Pi}_{ab}) . 
    \]
\end{proof}

\noindent Now we prove that the sequence of covers we use in the proof of Lemma 2 in the main text exists.

\begin{lemma}
    \label{app_auxlemma_covers}
    $\exists$ $\{B_k\}_{k=0}^K$ with $K < \infty$ of $\T{\Pi}_{ab}$ with $B_k \subset \T{\Pi}_{ab}$ s.t. for $k = 0,...,K$
    \begin{itemize}
        \item For all $\pi \in \T{\Pi}_{ab}$, there exists $b \in B_k$ such that $D_n(\pi,b) \leq 2^{-k}$,
        \item $|B_k| = N_{D_n}(2^{-k},\T{\Pi}_{ab},\{X_i,\Gamma_{i,\floor{n/2}+i}\}_{i=1}^{\floor{n/2}}) \leq |\T{\Pi}_{ab}|$.
    \end{itemize}
\end{lemma}

\begin{proof}
    Note that $|\T{\Pi}_{ab}| < 2^{\floor{n/2}} < \infty$ since $X_i$'s are fixed. Since $\T{\pi}_{ab}$ is finite and $B_k \subset \T{\Pi}_{ab}$ for all $k$, there is finite $K$ s.t. we can set $B_K = \T{\Pi}_{ab}$. This is because for any $B_k$ which is a strict subset of $\T{\Pi}_{ab}$ there is $\pi \in \T{\Pi}_{ab}$ s.t for all $b \in B_k$, $D_n(b,\pi) > a > 0$ and there is $K > 0$ s.t. $2^{-K} < a$. $K$ is finite since there are only finitely many subsets of $\T{\Pi}_{ab}$. For $B_{K-1}$ we look through all possible strict subsets for one which satisfies our conditions, if there is none we know that $B_{K-1} = \T{\Pi}_{ab}$ does satisfy them. In this way, we can go backwards and build the sequence of covers.
\end{proof}

\begin{lemma}
    \label{app_auxlemma_VC}
    $VC(\T{\Pi}_{ab}) \leq 2VC(\Pi) - 1$.
\end{lemma}
\begin{proof}
    Let $\pi_t(X_i) = \ind(\pi(X_i) = t)$ for $t \in \{0,1\}$. Define $\Pi_t = \{\ind(\pi(X_i) = t): \pi \in \Pi\}$. Note that $\Pi_1 = \Pi$ and that $VC(\Pi_0) = VC(\Pi_1)$ by Lemma 9.7 in \cite{kosorok2008introduction}. Now note that for any $(a,b) \in \{0,1\}^2$
    \[
      \T{\Pi}_{ab} = \{\pi_a \cdot \pi_b: (\pi_a,\pi_b) \in \Pi_a \times \Pi_b\}, 
    \]
    so Lemma 9.9 (ii) in \cite{kosorok2008introduction} yields the desired result.
\end{proof}

\subsection{IGM Corollary}

\begin{corollary}
    The bound in Theorem 1 applies to the IGM example. 
    \label{corollary_igm}
\end{corollary}

\begin{proof}[Proof of Corollary]
    Let $\Gamma_{ij}^{ab}$ and $\hat{\Gamma}_{ij,l}^{ab}$ be defined as in the IGM example and let
    \begin{align*}
        K(\pi) &= \Ex \biggl[ \sum_{(a,b) \in \{0,1\}^2} \Gamma_{ij}^{ab}\pi_{ab}(X_i,X_j)\biggr], \quad \T{K}_n(\pi) = \binom{n}{2}^{-1} \sum_{i<j} \biggl[ \sum_{(a,b) \in \{0,1\}^2} \Gamma_{ij}^{ab}\pi_{ab}(X_i,X_j)\biggr], \\
        \hat{K}_n(\pi) &= \binom{n}{2}^{-1} \sum_{l=1}^L \sum_{(i,j)\in I_l} \biggl[\sum_{(a,b) \in \pi}\hat{\Gamma}_{ij,l}^{ab}(Z_i,Z_j,\hat{\gamma}_l,\hat{\varphi}_l,\hat{\alpha}_l)\pi_{ab}(X_i,X_j).\biggr].
        \end{align*}
    Note also that $W(\pi) = - |K(\pi) - t|$. Hence, we can write the regret as
    \begin{align*}
        \Ex\biggl[&sup_{\pi \in \Pi_n} - |K(\pi) - t| + |K(\hat{\pi}) - t| \biggr] \leq \Ex\biggl[\sup_{\pi \in \Pi_n}|K(\pi) - K(\hat{\pi})|\biggr].
    \end{align*}
    The result follows from applying Theorem 1 with $W$ replaced by $K$.
\end{proof}

\bibliographystyle{ectabib}
\bibliography{references}

\end{document}